\documentclass{article}
\usepackage{authblk} % for author blocks
\usepackage{amsmath}
\usepackage{booktabs} % For formal tables

\usepackage{balance}

\usepackage{xcolor}
\usepackage[utf8]{inputenc}
\usepackage{bm}
\usepackage{amsfonts}
\usepackage{amsthm}
\usepackage{thmtools}
\usepackage{mathtools}
\usepackage{algorithm}
\usepackage{algorithmic}

\usepackage[refpage,prefix]{nomencl} \newcommand*{\nom}[3]{#2\nomenclature[#1]{#2}{#3}} \makenomenclature
\usepackage[draft=false,pageanchor]{hyperref} %fix-1
\newtheorem{example}{Example}
\newtheorem{theorem}{Theorem}
\newtheorem{conjecture}{Conjecture}

\newtheorem{lemma}{Lemma}
\newtheorem{definition}{Definition}
\newtheorem{mclaim}{Claim} %regular claims in llncs are un-numbered, so I define my own.
\newtheorem{cor}{Corollary}

\newcommand{\R}{\mathbb{R}}
\newcommand{\E}{\operatorname{\mathbb{E}}}
\newcommand{\N}{\mathbb{N}}
\newcommand{\etal}{\textit{et al}.}

\DeclarePairedDelimiter{\ceil}{\lceil}{\rceil}
\DeclarePairedDelimiter\floor{\lfloor}{\rfloor}

\DeclareMathOperator{\Var}{Var}

\begin{document}

\title{Redesigning Bitcoin's Fee Market}

\renewcommand\Affilfont{\itshape\small}
\setcounter{Maxaffil}{1}

\author[1]{Ron Lavi}
% \orcid{0000-0002-5215-5165}
\affil[1]{University of Bath, UK and Technion -- Israel Institute of Technology, Israel. ron.lavi.ac@gmail.com}
\author[2]{Or Sattath}
% \orcid{0000-0001-7567-3822}
\affil[2]{Ben-Gurion University of the Negev, Israel. sattath@post.bgu.ac.il}
\author[3]{Aviv Zohar}
% \orcid{0000-0001-8539-9222}
\affil[3]{The Hebrew University of Jerusalem, Israel. avivz@cs.huji.ac.il}

\date{}

\maketitle

\begin{abstract}
The Bitcoin payment system involves two agent types: Users that transact with the currency and pay fees and miners in charge of authorizing transactions and securing the system in return for these fees. Two of Bitcoin's challenges are (i) securing sufficient miner revenues as block rewards decrease, and (ii) alleviating the throughput limitation due to a small maximal block size cap. These issues are strongly related as increasing the maximal block size may decrease revenue due to Bitcoin's pay-your-bid approach. To decouple them, we analyze the ``monopolistic auction'' \cite{goldberg2006competitive}, showing: (i) its revenue does not decrease as the maximal block size increases, (ii) it is resilient to an untrusted auctioneer (the miner), and (iii) simplicity for transaction issuers (bidders), as the average gain from strategic bid shading (relative to bidding one's value) diminishes as the number of bids increases.
\end{abstract}

% \begin{CCSXML}
% <ccs2012>
% <concept>
% <concept_id>10003752.10010070.10010099.10010101</concept_id>
% <concept_desc>Theory of computation~Algorithmic mechanism design</concept_desc>
% <concept_significance>500</concept_significance>
% </concept>
% <concept>
% <concept_id>10003752.10010070.10010099.10010103</concept_id>
% <concept_desc>Theory of computation~Exact and approximate computation of equilibria</concept_desc>
% <concept_significance>500</concept_significance>
% </concept>
% <concept>
% <concept_id>10003752.10010070.10010099.10010107</concept_id>
% <concept_desc>Theory of computation~Computational pricing and auctions</concept_desc>
% <concept_significance>500</concept_significance>
% </concept>
% <concept>
% <concept_id>10010520.10010521.10010537.10010540</concept_id>
% <concept_desc>Computer systems organization~Peer-to-peer architectures</concept_desc>
% <concept_significance>500</concept_significance>
% </concept>
% </ccs2012>
% \end{CCSXML}

% \ccsdesc[500]{Theory of computation~Algorithmic mechanism design}
% \ccsdesc[500]{Theory of computation~Exact and approximate computation of equilibria}
% \ccsdesc[500]{Theory of computation~Computational pricing and auctions}
% \ccsdesc[500]{Computer systems organization~Peer-to-peer architectures}

% \keywords{Bitcoin; Cryptocurrency; Blockchain; Auction-Theory; Fee-market}

%https://bitcointalk.org/index.php?topic=1835475.0%3Ball

\section{Introduction}
Bitcoin's security relies on its ability to attract honest miners and to incentivize them to invest large amounts of computational power~\cite{nakamoto08bitcoin}. This is done by rewarding miners for the creation of new blocks. The payment comprises a block reward that is a fixed amount per block, and a fee that each transaction pays for its inclusion in the block. Since the block reward is cut in half approximately every 4 years, transaction fees gradually become the miners' main incentive to participate. The decision on which transactions to include in the block is made using a ``pay-your-bid'' auction: each transaction contains a fee for its inclusion in the block. A miner, in turn, includes the highest bidding transactions that fit in the block. In this auction, larger block sizes imply smaller transaction fees, since all transactions will aim to pay the lowest possible fee for their inclusion in the block. The Bitcoin block size therefore significantly influences miners' revenue. Indeed, one of the main arguments against a block-size increase in Bitcoin has been that it may cause transaction fees to drop significantly (see, e.g.,~\cite{houy2014economics}). 
Thus, two of the Bitcoin system's main problems, i.e., obtaining sufficient revenue for the miners as the block reward gradually decreases and the throughput limitation as a result of the maximal block size, are in fact strongly related. This connection, however, is a major obstacle to their resolution. Our motivation in this paper is to decouple the issues of revenue and  maximal block size (which should be determined according to other technical considerations, like block propagation times) and to provide scalability from the economic perspective of the fee market.

While most of auction theory assumes a trusted auctioneer that honestly follows the protocol, a miner (the auctioneer of the block) cannot be trusted. If adding fake bids can increase profit, a miner can add transactions by moving funds from one of its accounts to another. Similarly, a miner can ignore bids. An untrusted miner can manipulate a $2^{nd}$ price auction, for example, by adding a fake bid that is very close to the highest bid. Users will realize that and shade down their bids.
%A pay-your-bid auction is immune to such manipulations -- adding fake bids or removing true bids will not increase the auctioneer's revenue.
Though cryptography can, in general, help resolve some of these issues, as far as the authors are aware, cryptographic techniques are not applicable to Bitcoin: they assume a trusted setup, that the miner is known in advance, or multiple communication rounds between all parties as in secure multi-party computation~\cite{naor99privacy,lipmaa2002secure,abe2002m+,bogetoft2006practical}.

We analyze here a potential solution to these issues, the \emph{monopolistic auction} \cite{goldberg2006competitive}.\footnote{Termed the ``optimal single price omniscient auction'' in \cite{goldberg2006competitive}, this auction was suggested as a benchmark, and as far as we know, its game-theoretic properties have  not been analyzed in the Bitcoin context.} We observe that in the Bitcoin setting, the monopolistic auction is immune to untrusted miners, on the one hand, and it decouples the revenue issue from the maximal block size issue, on the other. In this auction, given bids $b_1,...,b_n$, the miner chooses which transactions to include in the block. All chosen transactions pay the same fee -- the lowest bid in the block (as must be verified by the protocol). The maximal revenue resulting from this method cannot be increased by adding fake bids or by ignoring true bids. Other manipulations like side payments between miners and users are also not beneficial. In fact, we are not aware of any beneficial miner manipulation that can be exploited in the frame of this auction. Additionally and in contrast to a pay-your-bid auction, increasing the maximal block size does not decrease the revenue of the monopolistic auction. Choosing a threshold price is conceptually similar to dynamically choosing the block size -- the miner chooses the size of the block it creates as a function of the bids received. In some bid instances, blocks will be small, while in others they will be large.

In a pay-your-bid auction, a user must strategically determine the lowest bid it can submit such that it will be included in the block (``bid shading''). Indeed, Bitcoin wallets use various fee estimation techniques. It is not clear how to do this optimally, and in any case, it requires some non-trivial computational effort (data gathering, statistical analysis, etc). Bid shading can also help in worst-case instances of the monopolistic auction, and as a result, \cite{goldberg2006competitive} do not analyze its strategic properties but suggest it only as a benchmark for other truthful auctions they define. These other truthful auctions from~\cite{goldberg2006competitive}  are less applicable vis-\'a-vis our needs (see Section~\ref{sec:summary}). Although the approach is not truthful, in the monopolistic auction, a user, rather than paying her bid, pays  some threshold value (a function of all bids) which w.h.p is only remotely related to one's bid. The main technical novelty of this paper is to show that the expected gain from optimal bid shading relative to truthfully bidding the user's value (which is her maximal willingness to pay) decreases to zero as the number of bids grows. On average, all users will benefit very little, if at all, from bid shading in the monopolistic auction. Since a near optimal strategy is to bid the user's value, bidding in the monopolistic auction is simpler. Section~\ref{sec:first_approach} defines the formal framework used and proves this statement, and Section~\ref{sec:empirical_evaluation} verifies it via simulations using synthetic and actual Bitcoin bid distributions. The greater simplicity of the monopolistic auction is another one of its advantages over the pay-your-bid auction.

Our work analyzes a single block creation, and  ignores ``patient users'', who are willing to wait for future blocks. These users might lower their bids if they anticipate less competition in the future. Two important issues that remain to be investigated in future research are (1) how patient users shade down their bids and (2) how non-myopic users who create persistent transactions affect the bid distribution. Note that these issues affect the revenue not only of the monopolistic auction, but also of the pay-your-bid auction. We do not know, however, which of the two auctions is more strongly affected. Our analytical results suggest that the monopolistic auction collects at least as much revenue from impatient users as Bitcoin's current mechanism. We discuss temporal considerations further in Section~\ref{sec:temporal_considerations}.

\vspace{2mm}

\noindent
{\bf Subsequent Work.} 
After the preprint version of this manuscript was published, Andrew Chi-Chih Yao~\cite{yao18incentive} further analyzed the monopolistic price as a fee mechanism for Bitcoin, in the process resolving our main conjecture, Conjecture~\ref{con:delta_goes_to_zero} -- see the discussion there.
In addition, Mark Friedenbach~\cite{friedenbach17rebatable,friedenbach18forward} proposed an implementation of our proposed fee market as a Bitcoin soft-fork (i.e., an upgrade which is compatible with older version of the Bitcoin client, see~\cite{narayanan16bitcoin} for details regarding a soft-fork).

Basu, Easley, O'Hara and Sirer~\cite{basu19towards} used a modified monopolistic price auction to determine inclusion in the block. They penalize miners if the blocks are not full, and additionally, they average the reward from blocks between different miners. Their goal is to maximize social welfare rather than the miners' revenue.
The main analytical difference between their work and ours is that theirs is based mostly on simulations while ours is proved for the most part formally. 
Perhaps surprisingly, even though their design goal and their mechanism is different than ours, in terms of abstract conclusions, our results agree with theirs: they also establish that the incentive of miners and users to manipulate diminishes as $n$ grows.

\vspace{2mm}

\noindent
{\bf Additional Related Work.} There are relatively few works devoted to analyses of the Bitcoin fee market and fewer still that suggest modifications to the basic mechanism. \cite{kroll2013economics} provided an early analysis of the economics of Bitcoin mining, including a game theoretic analysis of the incentives to mine on the longest chain. \cite{babaioff12bitcoin} considered the incentives of miners to distribute transactions to each other and designed ways to share the fees in exchange for distribution. \cite{carlsten16instability} explored the security of Bitcoin and the incentive to mine blocks properly when block rewards diminish and fees dominate the revenue of miners. They showed that variance in fees may undermine the security of the protocol and subject it to different forms of deviations and attacks. \cite{bonneau16why} explored related bribery attacks on the protocol that are paid for through promises of higher rewards for attackers that construct blocks off the longest chain. \cite{peter2015transaction} considered the removal of the block limit altogether, arguing that delays in large block propagation times, which, in turn, imply a higher likelihood that the block is abandoned (often referred to as an orphaned block), will cause miners to restrict their own block sizes. The paper analyzed the resultant fee market that emerges. \cite{huberman2017monopoly} modeled the transaction fee market and assumed that users benefit less if their transactions are delayed. They showed that such delays together with the congestion that may naturally occur in blocks due to queuing effects can lead to non-zero bids for transacting users even if blocks are not completely full.
As far as we are aware, no previous work has explored a different mechanism for the fee market in crypto-currencies. 

The paper \cite{azevedo18strategy} shows that a class of auctions that includes the monopolistic auction is ``strategy-proof in the large'', when values are i.i.d.~from some distribution with a finite support. Strategy-proofness in the large means that truthfulness is an epsilon BNE and epsilon goes to zero as the number of bidders goes to infinity. Our limit results on the monopolistic auction are similar, although we prove this for a different, stronger, notion -- the discount ratio. See a more detailed discussion in Section \ref{sec:first_approach}.

Two recent papers~\cite{leshno2019bitcoin, chen2019axiomatic} take the axiomatic approach to characterize good block reward schemes, i.e., schemes that satisfy sets of axioms, and show that the proportional selection rule is the unique selection rule satisfying several reasonable axioms.

%\noindent
%{\bf A full version} with greater detail and the full proofs is in Ref.~\cite{lavi17redesigning}.

Several surveys that discuss the state-of-art of integrating game theoretic considerations with cryptocurrencies are given, e.g., in~\cite{bohme15bitcoin,liu19survey,azouvi19tools}.

\paragraph{Terminology}
In general, we use Bitcoin’s terminology (e.g., users, transactions and miners). In some rare cases (especially when citing seminal auction theory results) we use the standard auction theory terminology instead (bidders, auctioneer, etc.).

\paragraph{Structure of this paper.} Section ~\ref{sec:first_approach} lays out the basic model, and further studies the properties of the monopolistic auction. Section ~\ref{sec:p_multibid} further generalizes the monopolistic auction to account for multiple strategic bids submitted by the same user. Section ~\ref{sec:empirical_evaluation} conducts an extensive empirical evaluation of our theoretical claims and conjectures. This is especially useful since our theoretical bounds holds in the large $n$ limit. Section ~\ref{sec:RSOP} conducts a theoretical and experimental analysis of the suitability of the RSOP auction in our setting. Section~\ref{sec:summary} gives concluding remarks and suggestions for future research directions. Appendices~\ref{ap:sp_well_formed} to~\ref{sec:analysis_RSOP_example} provide full proofs, and Appendix~\ref{sec:nomenclature} lists the main notations and symbols we use.

\section{The Monopolistic Auction}\label{sec:first_approach}

There are $n$ transactions, each of which has a privately known value $v_i$ for its inclusion in the block. This is the maximum fee that the transaction is willing to pay. The value can be deduced, e.g., by comparing to an alternative transaction cost via a bank, credit card, etc. 
Transactions submit bids, and based on these, the miner decides which transactions to include in the block. The bid vector is sorted so that \nom{b}{$\bm b$}{The bid vector $\bm b=(b_1,\ldots,b_n) $ satisfying $b_1\geq \ldots \geq b_n$}$=(b_1,\ldots,b_n)$ satisfies  $b_1 \geq \ldots \geq b_n$. Let \nom{k}{$k^*(\bm b)$}{The index $k$ that maximizes $k \cdot b_k$. In case of ties, $k^*$ is chosen to be the maximal such index} be the index $k$ that maximizes $k \cdot b_k$. In case of ties, $k^*$ is chosen to be the maximal such index. Define:
\begin{equation}
R(\bm b)\equiv k^*(\bm b) \cdot b_{k^*(\bm b)} \ , p^{M}(\bm b) \equiv b_{k^*(\bm b)},
\label{eq:monopolistic_revenue_and_price}
\end{equation}
where \nom{R}{$R(\bm b)$}{The monopolistic revenue, $R(\bm b)\equiv k^*(\bm b) \cdot b_{k^*(\bm b)}$.} is termed the ``monopolistic revenue'' and \nom{pm}{$p^{M}(\bm  b)$}{The monopolistic price} is termed the monopolistic price. In the monopolistic auction, the miner includes in the block the $k^*(\bm b)$ transactions who submitted the highest bids and charges a payment of $p^{M}(\bm  b) = b_{k^*(\bm b)}$ from each of these transactions. The miner's revenue from this block is therefore $R(\bm b)$.

An untrusted miner that can add fake bids does not have any incentive to do so in this auction. Specifically, the Bitcoin protocol can verify that all transactions in the same block pay the same price. However, it cannot prevent miners from adding fake bids by exploiting transactions that move funds between two accounts of the same miner. These fake bids can affect the payments of the ``real'' bids and increase the revenue of the miner, as is the case in a second-price auction. In the monopolistic auction, in contrast, the addition of fake bids cannot increase the revenue: given any vector of real bids $\bm{b}$, since all transactions in the block pay the same price, the maximal revenue is $R(\bm b)$ if the miner does not create another block in the near future. However, a user might benefit from submitting $b_i < v_i$, as demonstrated by the following examples:

\begin{example} %simple shading down by not much
Suppose $n$ users that have the same value $v_i = 1$. If all bids are $b_i = v_i$, then $R(\bm b) = n, k^*(\bm b) = n$, and $p^{M}(\bm  b) = 1$, i.e., all transactions get accepted to the block and they all pay $1$. A strategic user, say $1$, can reduce her payment by reducing her bid to $b_1 = \frac{n-1}{n}$. With this bid, the monopolistic price decreases to $\frac{n-1}{n}$: if all transactions that bid $1$ are accepted, the revenue is $n-1$, and if all transactions that bid at least $b_1 = \frac{n-1}{n}$ are accepted, the revenue is still $n-1$.
\end{example}

The gain from bid shading in this example vanishes as the number of transactions (and the block size) increases. In other cases, gains do not vanish even when the number of transactions grows:

\begin{example} %shading down can be quite significant
\label{ex:large_discount}
$v_1 = \ldots = v_{\frac{n}{2}+1}=2, v_{\frac{n}{2}+2} = \ldots = v_n = 1$, given some even $n$. If user $1$ is truthful and bids $2$, her transaction will get accepted to the block and the price will be $2$. A careful look reveals that she can reduce her price to $1$ by bidding $1$ instead.
\end{example}

However, if the values for this example are taken from realistic distributions, then the probability that it will be realized will usually be very small.

Therefore, rather than allege that there are no profitable deviations from truth-telling, we will only claim that the expected gain from any non-truthful bid diminishes as $n$ grows. For example, if the distribution of values has a finite support, bidding the next lower value in the support {\em weakly dominates} truth-telling. More pronounced bid shading may be beneficial; We will show, however, that the expected gain from such a deviation or from any other deviation becomes extremely small as $n$ grows. Therefore, we believe that it makes sense to conclude that no deviations will likely happen.

\begin{example} \label{example::equal_rev}
Consider the distribution $F$ with finite support: $\Pr(v=1)=0.5, \Pr(v=2)=0.5$. For simplicity, assume that users can only submit bids of 1 or 2. Every user with $v=2$ has a weakly dominant strategy to bid 1, since she derives 0 utility when the monopoly price is set to be 2, and by bidding 1 she increases the chance that the monopoly price will be 1. However, the probability that bidding 1 instead of 2 will actually make a difference is vanishingly small. Specifically, a user with $v=2$ can improve the outcome (i.e., change the empirical monopoly price from 2 to 1) by bidding b=1, when the number of users is $2n$, only in the event that there are exactly $n+1$ users with value equals to 2. The probability of this event is $\frac{\binom{2n}{n+1}}{2^{2n}} = \Theta\left(\frac{1}{\sqrt{n}}\right)$. We note that this is the most extreme example -- as the support size increases, the bid shading suggested in this example becomes negligible.
\end{example}

We validate this intuition more generally via both theoretical and empirical analyses. First we define some terms. Fix a user $i$ and a vector of bids $\bm b_{-i}\equiv (b_1,\ldots, b_{i-1},b_{i+1},\ldots,b_n)$ of the other users. Define the ``strategic price'':
\begin{equation}
p^{S}(\bm b_{-i}) \equiv \\
\inf \{b_i \in \R\ |\ p^{M}(b_i,\bm b_{-i}) \leq b_i  \} = \min \{b_i \in \R\ |\ p^{M}(b_i,\bm b_{-i}) \leq b_i  \} .    
\label{eq:ps_well_defined}
\end{equation}

The equality is proven in Appendix~\ref{ap:sp_well_formed}.
In words, \nom{ps}{$p^{S}(\bm b_{-i})$}{The strategic price: The lowest possible bid for $i$ to be included in the block, when the other bids are $\bm b_{-i}$} is the lowest possible bid for $i$ to be included in the block, when the other bids are $\bm b_{-i}$.
%\footnote{The minimum is well defined by Claim~\ref{cl:f_property} in the Appendix, using $f(b_i) = p^{M}(b_i,b_{-i})$. Claim~\ref{cl:p_monopolistic_satisfies_the_f_property}, shows that $f(b_i)$ satisfies the property required by Claim~\ref{cl:f_property}, and clearly, there exists $b_i$ such that $b_i \geq f(b_i)$, e.g., taking any $b_i > \max_{j \neq i} b_j$.}
Figure~\ref{fig:p_mon} gives an example. This definition does not capture a case where a user can provide false-name bids (i.e., can split her transaction into several transactions), which is discussed in Section~\ref{sec:p_multibid}.

\begin{figure}[t!]
\centering \includegraphics[width=\columnwidth]{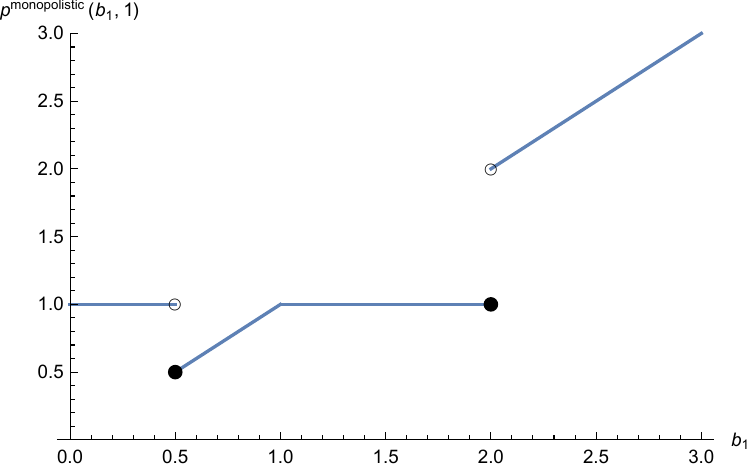}
\caption{The function $p^{M}(b_1,1)$ as a function of $b_1$, demonstrating that $p^{M}$ is neither monotone nor continuous. As can be seen, the strategic price of the first user is $0.5$.}
\label{fig:p_mon}
\end{figure}

Figure~\ref{fig:p_strategic} presents the strategic prices in an example with $n=256$ users. In this example, $p^{S}$ is monotonically decreasing among winning users (higher winning bids imply lower strategic prices) but not among all users. This property is formally proved in Claim~\ref{cl:mon_for_win_st_p} in Appendix~\ref{sec:auxiliary_claims}. Additionally, here and in most other simulations, almost all winning bids have the same strategic price, the exception being the very few lowest winning bids (i.e., those with prices slightly higher than the monopolistic price). The strategic price of these bids is higher, i.e., they gain less from being strategic.
%Similarly, most losing bids have the same strategic price, except for the few highest ones.

While $p^{S}(\bm b_{-i})$ is the minimal price that user $i$ must pay to be included in the block, she will pay $p^{M}(v_i,\bm b_{-i})$ if she bids her true value. This possible gain from bid shading is captured by the notion of the ``discount ratio'':
\[
\delta_i(v_i,\bm b_{-i}) \equiv
\begin{cases}
1 - \frac{p^{S}(\bm b_{-i})}{p^{M}(v_i,\bm b_{-i})} & \text{if } v_i\geq p^{S}(\bm b_{-i})  \\
0 & \text{otherwise.}
\end{cases}
\label{eq:delta_i}
\]
Clearly, $0 \leq$ \nom{deltai}{$\delta_i(v_i,\bm b_{-i})$}{The discount ratio} $\leq 1$. If $v_i < p^{S}(\bm b_{-i})$, every bid of at most $v_i$ loses. Hence, the gain from bid shading is $0$. If $v_i \geq p^{S}(\bm b_{-i})$, user $i$ can win and pay $p^{S}(\bm b_{-i})$. With an honest bid, she will pay $p^{M}(v_i,\bm b_{-i})$. She can thus save a percentage of $1 - \frac{p^{S}(\bm b_{-i})}{p^{M}(v_i,\bm b_{-i})}$ from her price by strategic bid shading.

Example~\ref{ex:large_discount} above shows that the discount factor does not vanish in the worst case, even as $n$ grows. In this example, the discount factor is $\frac{1}{2}$ for all possible values of $n$. We therefore conduct an average-case analysis. Assume that all true values are drawn i.i.d.~from some distribution \nom{F}{$F$}{All true values are drawn i.i.d.~from some distribution $F$ on $\mathbb{R}_{>0}$} on $\mathbb{R}_{>0}$. The average discount ratio is then
\nom{deltaaverage}{$\Delta_n^{average}$}{$ =  \E_{ (v_1,\ldots,v_n) \sim F} [\delta_1(v_1,\bm v_{-1})]$}$ =  \E_{ (v_1,\ldots,v_n) \sim F} [\delta_1(v_1,\bm v_{-1})]$ (the choice of user $1$ is arbitrary since all are symmetric a-priori). Note that this definition implicitly assumes that $\bm b_{-1} = \bm v_{-1}$. This is similar to the usual logic in equilibrium analysis -- assume all others are truthful and show that user $i$ does not have a reason to deviate from truthfulness as well. Indeed, we will show that all discount ratios go to zero as $n$ goes to infinity. We therefore rely on this logic in all our definitions.

We also consider two stronger notions. In the first, for each realization $v_1,...,v_n$, we consider the maximal discount ratio among all users:
\[
\delta_{max}(\bm v) \equiv \max_{i} \delta_i(v_i, \bm v_{-i})\ ,\ \Delta_n^{max} \equiv \E_{ (v_1,\ldots,v_n) \sim F}  [  \delta_{max}( \bm v)]
\]
Clearly, for all $n$, \nom{deltanmax}{$\Delta_n^{max}$}{$\Delta_n^{max} \equiv \E_{ (v_1,\ldots,v_n) \sim F}  [  \delta_{max}( \bm v)$} $\geq \Delta_n^{average}$  since
for every $\bm v$ and every $i$, \nom{deltamax}{$\delta_{max}(\bm v)$}{$\delta_{max}(\bm v) \equiv \max_{i} \delta_i(v_i, \bm v_{-i})$} $\geq \delta_i(v_i, \bm v_{-i})$.

\begin{figure}[t!]
\centering \includegraphics[width=\columnwidth]{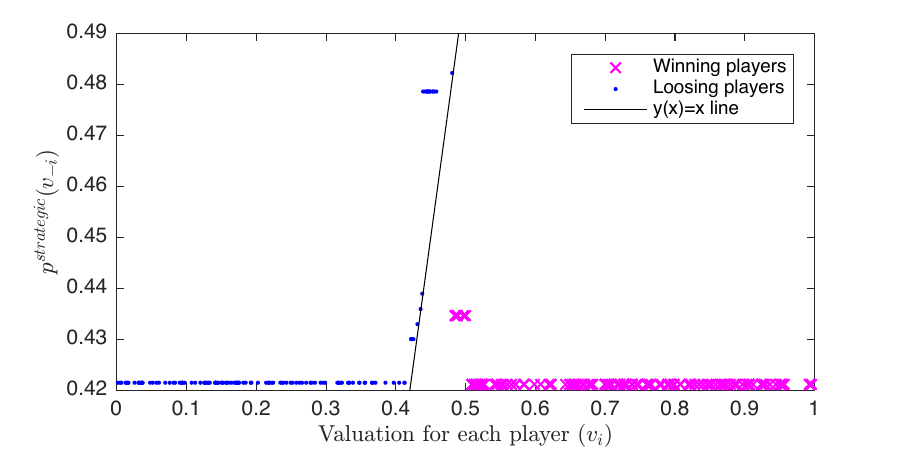}
\caption{An example with $n=256$ bids sampled i.i.d. from the uniform distribution on $[0,1]$. For each user, the x-axis shows her values and the y-axis shows her strategic price. Winning bids, which are marked with ``x'', must be at least as high as their respective strategic price, i.e., to the right of the black line.}
\label{fig:p_strategic}
\end{figure}

\begin{theorem}
For any distribution $F$ with a finite support size,
$\lim_{n \rightarrow \infty}\Delta_n^{max} = 0$.
\label{thm:delta_max_negligible}
\end{theorem}

\noindent
Section~\ref{sec:proof-of-thm-1} proves this theorem.

In the second extension, defined only for distributions with a finite bounded support, we fix player 1 and assume that she deterministically has a value that maximizes her discount ratio (the worst possible value for her, from our perspective) while all other values are probabilistic. This depicts a situation commonly assumed in game theory wherein a user knows her own value and (only) the distribution of the other values. Let \nom{vmax}{$v_{max}$}{$v_{max} \equiv \max \mbox{Support(F)}$} $\equiv \max \mbox{Support(F)}$.
\begin{equation}
  \delta_n^{GT}(\bm b_{-1}) \equiv \max_{v_1 \in \mbox{Support}(F)} \delta_1(v_1, \bm b_{-1}) \ \ \ = \delta_1(v_{max}, \bm b_{-1}).
\end{equation}
The equality is proven in Claim \ref{cl:delta_is_monotone}.
\begin{equation}
  \Delta_n^{GT} \equiv  \E_{\bm (v_2,\ldots,v_n)\sim F} [\delta_n^{GT}(\bm v_{-1})]  
\end{equation}

\begin{theorem}
\label{thm:delta_max_gt}
%Fix any distribution $F$. If $F$ has a finite support size, then,
For any distribution $F$ with a finite support size,
$\lim_{n \rightarrow \infty}$ \nom{Deltangt}{$\Delta_n^{GT}$}{$\Delta_n^{GT} \equiv  \E_{\bm (v_2,\ldots,v_n)\sim F} [\delta_n^{GT}(\bm v_{-1})]  $} $= 0$.
\end{theorem}

\noindent
The proof of this theorem is very similar to that of Theorem~\ref{thm:delta_max_negligible}, as explained in Section~\ref{sec:proof-of-thm-1}.

Our empirical analysis focuses on $\Delta_n^{max}$, since it is well-defined even for distributions with unbounded support.

While $\lim_{n \rightarrow \infty} | \Delta_n^{GT} - \Delta_n^{max} | = 0$, the relation between $\Delta_n^{GT}$ and $\Delta_n^{max}$ depends on the distribution, i.e., one term is not always larger than the other.
For example, consider the distribution $\Pr(v_i = 1) = p$ and $\Pr(v_i = \epsilon) = 1-p$, where $0 < \epsilon < 1$. In this case, if $p<0.5$, then $\Delta_{n=2}^{GT} > \Delta_{n=2}^{max}$, and if $p>0.5$, then $\Delta_{n=2}^{GT} < \Delta_{n=2}^{max}$, since:
\begin{align*}
\Delta_{n=2}^{GT} =&\E_{v_2\sim F}[\delta_1(1,v_2)]= p\delta(1,1)+(1-p)\delta(1,\epsilon) = p\cdot \frac{1}{2} + (1-p)\cdot(1-\frac{\epsilon}{2})\\
%\end{align*}
%\begin{align*}
\Delta_{n=2}^{max} =& \E_{(v_1,v_2)\sim F}[\max_{i} \delta_i(v_1,v_2)]\\
=& \ p^2 \delta_1(1,1) + (1-p)^2\delta_1(\epsilon,\epsilon) +2p(1-p) \max_i \delta_i(1,\epsilon)\\
%\stackrel{\text{Claim }\ref{cl:delta_bigger_for_larger_users}}{=}
=& \ p^2 \delta_1(1,1) + (1-p)^2\delta_1(\epsilon,\epsilon) +2p(1-p) \delta_1(1,\epsilon)\\
=& \ (1-2p(1-p))\cdot \frac{1}{2} + 2p(1-p)\cdot(1-\frac{\epsilon}{2})
\end{align*}

\noindent
We believe that our theorems can be further generalized:
\begin{conjecture}[Nearly Bayesian-Nash Incentive-Compatibility]
~~
\begin{enumerate}
\item For any distribution $F$, $\lim_{n\rightarrow \infty}\Delta_n^{average} = 0$. In particular, $\Delta_n^{average}=O\left(\frac{1}{\sqrt n}\right)$, where the constant in the $O(\cdot)$ notation may depend on $F$.\label{it:delta_n_average_zero}

\item If $F$ has a bounded support (still possibly with an infinite cardinality) $\lim_{n\rightarrow \infty}\Delta^{\max}_n = 0$. In particular, $\Delta^{\max}_n = O\left(\frac{1}{\sqrt n}\right)$. The same holds for $\Delta_n^{GT}$.
\label{it:delta_n_max_zero}

\item There exists a distribution $F$ with unbounded support such that\\ $\lim_{n\rightarrow \infty}\Delta^{\max}_n > 0$. The same holds for $\Delta_n^{GT}$.\label{it:delta_n_max_non_zero}
\end{enumerate}
\label{con:delta_goes_to_zero}
\end{conjecture}

\noindent
The $O\left( \frac{1}{\sqrt n}\right)$ in the above conjecture can be achieved by Example~\ref{example::equal_rev}. 

Recently, Andrew Chi-Chih Yao resolved many issues that were left open in this work. Indeed, ``The main purpose [of Yao's paper] is to settle the Nearly IC Conjecture for the monopolistic price in the positive"~\cite{yao18incentive}.
Specifically, the first sentences in items~\ref{it:delta_n_average_zero} and~\ref{it:delta_n_max_zero} in the conjecture above were proved, and a slightly weaker claim was shown than the ``in particular" claims in these items. This was proved also when accounting for multiple strategic bids -- see Section~\ref{sec:p_multibid}. Item~\ref{it:delta_n_max_non_zero} was proved for the inverse distribution (see Eq.~\eqref{eq:inverse_distribution}), which was our primary candidate for this conjecture based on numerical data -- see Section~\ref{sec:empirical_evaluation}. 

\vspace{2mm}

\noindent
{\bf Game Theoretic Equilibrium Analysis. } In a Bayesian-Nash equilibrium (BNE), user $i$'s utility is $u_i(v_i, b_i, \bm b_{-i}) = v_i - p^M(\bm b)$ if $b_i \geq p^M(\bm b)$ and $0$ otherwise. A strategy $s_i: \R^+ \rightarrow \R^+$ gives $i$'s bid $s_i(v_i)$. $(s_1(\cdot),...,s_n(\cdot))$ is an \nom{epsilonBNE}{$\epsilon$-BNE}{An $\epsilon$-Bayesian Nash Equilibrium} if $\forall i$ and $v_i,b_i\in Support(F)$: $\E_{\bm v_{-i}}  [u_i(v_i, s_i(v_i), \bm v_{-i})] \geq \E_{\bm v_{-i}}  [u_i(v_i, b_i, \bm v_{-i})] -\epsilon$. The ``truthful'' strategy $s_i(v_i)=v_i$ is an $\epsilon$-BNE if for all $i$,
\begin{equation}
\label{eq:BNE}
 \max_{v_i,b_i \in Support(F)} \E_{\bm v_{-i}}  [u_i(v_i, b_i, \bm v_{-i}) - u_i(v_i, v_i, \bm v_{-i})]  \leq \epsilon
\end{equation}
Our discount-ratio analysis is similar but formally different. In fact, our approach is stronger.

\begin{theorem}
\label{thm:bne-vs-discount-ratio}
Assume that $F$ is bounded, with $v_{max}$ being the maximal point on the support of $F$. For any $n$, $s_i(v_i) = v_i$ is an $\epsilon$-BNE with $\epsilon = v_{max} \cdot \Delta^{GT}_n$ in the monopolistic auction with $n$ users.
\end{theorem}

\noindent
The proof is given in
 Appendix~\ref{sec:proof-of-thm-bne-vs-discount-ratio}.

If $F$ has a finite support, Theorem~\ref{thm:delta_max_gt} states that $\lim_{n \rightarrow \infty} \Delta^{GT}_n = 0$. Therefore, there exists a function $N(\epsilon)$ such that, for every $\epsilon > 0, n \geq N(\epsilon)$, truthfulness is an $\epsilon$-BNE in the game with $n$ users.
The opposite direction is not true; if truthfulness is an $\epsilon$-BNE, this does not imply that the discount ratio is at most $\epsilon$.
%The other direction, that truthfulness is $\epsilon$-BNE, which implies that the discount ratio is at most $\epsilon$, is not true.
This is because the deviation $b_i$ in the BNE analysis is fixed before taking the expectation, while in our analysis $b_i$ is a best response to a specific $\bm v_{-i}$ and we take the expectation over these possibly different best responses. In this respect, the discount-ratio analysis is stronger.

\subsection{Proof of Theorem~\ref{thm:delta_max_negligible}}
\label{sec:proof-of-thm-1}

Recall that $F$ is the distribution over true values. Given $\bm v$, define:\\ \nom{istar}{$i^*$}{$i^* \equiv \mbox{argmax}_{i=1,...,n} v_i$}$ \equiv \mbox{argmax}_{i=1,...,n} v_i$, \nom{kstar}{$k^*$}{$k^* \equiv k^*(\bm v_{-i^*})$} $\equiv k^*(\bm v_{-i^*})$, and \nom{pS}{$p^{S}$}{$p^{S} \equiv p^{S}(\bm v_{-i^*})$} $\equiv p^{S}(\bm v_{-i^*})$ (`S' stands for `Strategic').

\begin{lemma} For any $F$ with finite support size, there exists a constant $c>0$ (which may depend on $F$ but not on $n$) s.t.
$\lim_{n\rightarrow \infty }\Pr(k^* < c \cdot n) = 0$.
\label{le:small_block}
\end{lemma}
\begin{proof}
Let $v_{max}=\max \mbox{Support}(F)$, $k_{max}$ be the number of users in $\bm v$ whose value is $v_{max}$, and $p_{max}=\Pr_{v_i \sim F}[v_i=v_{max}]$. By linearity of expectation, $ \E[k_{max}] = p_{max}\cdot n$. A Chernoff bound implies that $\lim_{n\rightarrow \infty} \Pr\left(k_{max} < \frac{9 n p_{max}}{10}\right) = 0$. Users who bid $v_{max}$ win; therefore, $k^* +1 \geq k_{max}$. Thus,\\ $\lim_{n\rightarrow \infty} \Pr\left(k^* +1 < \frac{9 n p_{max}}{10} \right) = 0$.
Choosing $c=\frac{8 p_{max}}{10}$ completes the proof of the lemma.
\end{proof}

\noindent
Appendix~\ref{sec:auxiliary_claims} shows three useful properties:
\begin{restatable}{property}{pstrategicalternativedef}
$\forall \bm v,i$, $p^{M}(p^{S}(\bm v_{-i}),\bm v_{-i}) = p^{S}(\bm v_{-i})$.
\label{cl:p_strategic_alternative_def}
\end{restatable}
\begin{restatable}{property}{xequaly}Let $y$ be the smallest element in the support of $F$ that is at least $p^{S}$. Then, $y=p^{M}(\bm v)$ implies $p^{S} \geq \frac{k^*}{k^*+1} \cdot p^{M}(\bm v)$.
\label{cl:x_equal_y}
\end{restatable}
\begin{restatable}{property}{deltabiggerforlargerbidders}$\forall \bm v,i,j$: $v_i \geq v_j$ implies $\delta_i(v_i,\bm v_{-i}) \geq \delta_j(v_j,\bm v_{-j})$. 
\label{cl:delta_bigger_for_larger_bidders}
\end{restatable}

%These are claims~\ref{cl:p_strategic_alternative_def},~\ref{cl:delta_bigger_for_larger_users},~\ref{cl:x_equal_y} in Appendix~\ref{sec:auxiliary_claims}.

\iffalse

\begin{mclaim}
Let $F$ be any distribution with a finite support size and $\bm v = (v_1,..v_n)$ be $n$ i.i.d.~draws from $F$. Define \nom{num}{$\mbox{num}(\bm v, z)$}{$\mbox{num}(\bm v, z) \equiv |\{v_i|v_i\geq z \}|$} $\equiv |\{v_i|v_i\geq z \}|$ and random variables $n_z = \mbox{num}(\bm v,z)$ for any real number $z$. Then, for any arbitrary two $x,y\in \mbox{Support}(F)$ such that $x>y$,
$$
\lim_{n \rightarrow \infty} \Pr(n_x \cdot x > n_y \cdot y \geq (n_x - 1)x ) = 0.
%\label{eq:pr_h_g}
$$
\label{cl:B_has_zero_probability}
\end{mclaim}
\begin{proof}
The term $n_x \cdot x > n_y \cdot y \geq  (n_x - 1)x$ is the same as $n_x > \frac{y}{x} n_y \geq n_x - 1$, which is true if and only if $n_x = \lfloor{\frac{y}{x}n_y+1}\rfloor$. The triple $(n_x, n_y - n_x, n - n_y)$ is a trinomial distribution with probabilities $p_1 = \Pr_{v_i \sim F}(v_i \geq x), p_2 = \Pr_{v_i\sim F}(x > v_i \geq y), p_3 = \Pr_{v_i\sim F}(v_i < y)$ ($p_i$ depends on $F$ but not on $n$). Denoting $f(n_y) = \floor{\frac{y}{x}n_y+1}$, we conclude that $\Pr(n_x=f(n_y)) = \Pr(n_x \cdot x > n_y \cdot y \geq (n_x - 1)x)$.
Claim~\ref{cl:trinomial} therefore implies the current claim.
\end{proof}

\fi

\noindent
With these, we prove a second lemma:

\begin{lemma}
$\lim_{n\rightarrow \infty} \Pr(p^{S} < \frac{k^*}{k^*+1}p^{M}(\bm v))=0$.
\label{le:small_p_strategic}
\end{lemma}
\begin{proof}
Define $A$ as the event where $p^{S} < \frac{k^*}{k^*+1}p^{M}(\bm v)$. The proof defines an event $B$ s.t.~(i) $\lim_{n\rightarrow \infty} \Pr(B)=0$, and (ii) $A \subseteq B$. This implies $\lim_{n\rightarrow \infty} \Pr(A)=0$ as claimed.

%To define event $B$, fix any arbitrary two $x,y\in \mbox{Support}(F)$ such that $x>y$.
Define \nom{num}{$\mbox{num}(\bm v, z)$}{$\mbox{num}(\bm v, z) \equiv |\{v_i|v_i\geq z \}|$} $\equiv |\{v_i|v_i\geq z \}|$ and the random variables $n_z = \mbox{num}(\bm v,z)$ for any real number $z$.
%Let $h(x) = n_x \cdot x$, i.e., $h(x)$ is the revenue from $\bm v$ if the monopolistic price is $x$, and define $g(x) = (n_x - 1) \cdot x$. We first aim to bound the probability that $h(x)>h(y)\geq g(x)$. This is the same as the probability that $n_x > \frac{y}{x} n_y \geq n_x - 1$, which is true if and only if $n_x = \lfloor{\frac{y}{x}n_y+1}\rfloor$. The triple $(n_x, n_y - n_x, n - n_y)$ is a trinomial distribution with probabilities $p_1 = \Pr_{v_i \sim F}(v_i \geq x), p_2 = \Pr_{v_i\sim F}(x > v_i \geq y), p_3 = \Pr_{v_i\sim F}(v_i < y)$ ($p_i$ depends on $F$ but not on $n$). Denoting $f(n_y) = \floor{\frac{y}{x}n_y+1}$, we conclude that $\Pr(n_x=f(n_y)) = \Pr(h(x)>h(y)\geq g(x))$. Claim~\ref{cl:trinomial} in Appendix~\ref{ap:binomial_trinomial} therefore implies%, taking $f(n_y) = \floor{\frac{y}{x}n_y+1}$, that
%
%\vspace{-4mm}
%
%\begin{equation}
%\lim_{n \rightarrow \infty} %\Pr(h(x)>h(y)\geq g(x)) = 0.
%\label{eq:pr_h_g}
%\end{equation}
%
Let $B$ be the event in which $\exists x,y \in \mbox{Support}(F) \mbox{ s.t. } x>y \mbox{ and } n_x \cdot x > n_y \cdot y \geq (n_x - 1)x$. Since the support is of finite size, there is a fixed number of such pairs $x>y\in \mbox{Support}(F)$. By Claim~\ref{cl:B_has_zero_probability} in Appendix~\ref{ap:binomial_trinomial} and the union bound over these finite set of pairs, we conclude that $\lim_{n \rightarrow \infty} \Pr(B) = 0$. We show that $A \subseteq B$, i.e., $p^{S} < \frac{k^*}{k^*+1}p^{M}(\bm v)$ implies $\exists x>y \in \mbox{Support}(F) \mbox{ s.t. } n_x \cdot x > n_y \cdot y \geq (n_x - 1)x$. Let $x=p^{M}(\bm v)$ and let $y$ be the smallest element in the support of $F$ which is at least $p^{S}$. By Property~\ref{cl:x_equal_y}, the event $p^{S} < \frac{k^*}{k^*+1}p^{M}(\bm v)$ implies $x>y$. Since $p^{M}(v_{i^*},\bm v_{-i^*})=x$, it follows that $n_x \cdot x > n_y \cdot y$. Furthermore,

\vspace{-4mm}

\begin{align*}
n_y \cdot y \geq n_{p^{S}} \cdot p^{S} = \mbox{num}(\bm v, p^{S})  \cdot p^{S} &= \\
\mbox{num}((p^{S},\bm v_{-i^*}), p^{S}) \cdot p^{S}
&\geq \\ 
\mbox{num}((p^{S}, \bm v_{-i^*}), x) \cdot x &=(n_x-1)x
\end{align*}
where the first step follows since $n_y = n_{p^{S}}$ and $p^{S} \leq y$, the third step follows since
$v_{i^*} \geq x > y \geq p^{S}$, the fourth step follows from
Property~\ref{cl:p_strategic_alternative_def} that shows that
$p^{M}(p^{S}, \bm v_{-i^*}) = p^{S}$,
and the fifth step follows since $v_{i^*}\geq x > y \geq p^{S}$.

\iffalse
Let $B$ be the event in which $\exists x,y \in \mbox{Support}(F) \mbox{ s.t. } x>y \mbox{ and } h(x) > h(y) \geq g(x)$. Since the support is of finite size, there is a fixed number of such pairs $x>y\in \mbox{Support}(F)$. By Eq.~\eqref{eq:pr_h_g} and the union bound over these finite set of pairs, we conclude that $\lim_{n \rightarrow \infty} \Pr(B) = 0$. We show that $A \subseteq B$, i.e., $p^{S} < \frac{k^*}{k^*+1}p^{M}(\bm v)$ implies $\exists x>y \in \mbox{Support}(F) \mbox{ s.t. } h(x) > h(y) \geq g(x)$. Let $x=p^{M}(\bm v)$ and let $y$ be the smallest element in the support of $F$, which is at least $p^{S}$. By Property~\ref{cl:x_equal_y}, the event $p^{S} < \frac{k^*}{k^*+1}p^{M}(\bm v)$ implies $x>y$. Since $p^{M}(v_{i^*},\bm v_{-i^*})=x$, it follows that $h(x) = n_x \cdot x > n_y \cdot y = h(y)$. Furthermore,

%\vspace{-4mm}

\begin{align*}
h(y)& \geq h(p^{S}) =n_{p^{S}} \cdot p^{S} = \mbox{num}(\bm v, p^{S})  \cdot p^{S} =\mbox{num}((p^{S},\bm v_{-i^*}), p^{S}) \cdot p^{S}\\
&\geq \mbox{num}((p^{S}, \bm v_{-i^*}), x) \cdot x =(n_x-1) \cdot x = g(x)
\end{align*}
where the first step follows since $n_y = n_{p^{S}}$ and $p^{S} \leq y$, the fourth step follows since
$v_{i^*} \geq x > y \geq p^{S}$, the fifth step follows from
Property~\ref{cl:p_strategic_alternative_def} that shows that
$p^{M}(p^{S}, \bm v_{-i^*}) = p^{S}$,
and the sixth step follows since $v_{i^*}\geq x > y \geq p^{S}$.

\fi

\end{proof}

Let the ``bad'' event $E_1$ be the case where $k^* < c \cdot n$ ($c$ is taken from Lemma~\ref{le:small_block}) and the ``bad'' event $E_2$ be the case where  $p^{S} < \frac{k^*}{k^*+1}p^{M}(\bm v)$. By Property~\ref{cl:delta_bigger_for_larger_bidders}, $\delta_{max}(\bm v) = \delta_{i^*}(\bm v)$. Therefore, $\Delta^{max}_n=\E_{\bm v}[\delta_{i^*}(\bm v)]$. If $E_2$ does not hold then $p^{S} \geq \frac{k^*}{k^*+1} p^{M}(\bm v)$ and therefore $\delta_{i^*}(\bm v) \leq \frac{1}{k^*+1}$. To conclude:

%$lim_{n \rightarrow \infty} \Delta_n^{max} = \lim_{n \rightarrow \infty} \Pr[E_1\cup E_2]\E_v[\delta_{i^*}(v_{i^*},v_{-i^*})|E_1 \cup E_2] +$ $\Pr[E_1^c \cap E_2^c] \E_v\left[\delta_{i^*}(v_{i^*},v_{-1})|E_1^c \cap E_2^c\right]$\\ $\leq \lim_{n \rightarrow \infty}  \left[ \Pr[E_1] + \Pr[E_2] + \Pr[E_1^c \cap E_2^c] \cdot \E_v[\frac{1}{k^*+1}|E_1^c \cap E_2^c] \right] \leq$ \\
%$\lim_{n \rightarrow \infty} \left [ \Pr[E_1] + \Pr[E_2] + \Pr[E_1^c \cap E_2^c] \cdot \frac {1}{c\cdot n} \right] \leq$ $\lim_{n \rightarrow \infty}  \left[ \Pr[E_1] + \Pr[E_2] +  \frac {1}{c\cdot n} \right] =0$, implying Theorem~\ref{thm:delta_max_negligible}.

\begin{align*}
\lim_{n \rightarrow \infty} \Delta_n^{max} =&\lim_{n \rightarrow \infty}  \Big[ \Pr(E_1\cup E_2)\E_v[\delta_{i^*}(v_{i^*},\bm v_{-i^*})|E_1 \cup E_2]\\ 
& + \Pr(E_1^c \cap E_2^c) \E_v\left[\delta_{i^*}(v_{i^*},\bm v_{-1})|E_1^c \cap E_2^c\right] \Big] \\
\leq& \lim_{n \rightarrow \infty}  \left[ \Pr(E_1) + \Pr(E_2) + \Pr(E_1^c \cap E_2^c) \cdot \E_v[\frac{1}{k^*+1}|E_1^c \cap E_2^c] \right]\\
\leq& \lim_{n \rightarrow \infty} \left [ \Pr(E_1) + \Pr(E_2) + \Pr(E_1^c \cap E_2^c) \cdot \frac {1}{c\cdot n} \right]\\
\leq& \lim_{n \rightarrow \infty}  \left[ \Pr(E_1) + \Pr(E_2) +  \frac {1}{c\cdot n} \right] = 0,
\end{align*}

\noindent
implying Theorem~\ref{thm:delta_max_negligible}.
The proof of Theorem~\ref{thm:delta_max_gt} is very similar.
The main difference between the two theorems is that in the latter, we need to set the value of the first player to the maximal value in the support of $F$. All statements of the above proof then hold with respect to the first player, instead of player $i^*$.

\section{Multiple strategic bids}
\label{sec:p_multibid}
In some cases it is beneficial to split one's bid to several separate transactions with several separate bids. In fact, such a strategy sometimes enables a losing transaction to be included in the block, as the following example demonstrates:

\begin{example} 
Let $\bm v = (5, 2, 1, 1)$. With bids $\bm b = \bm v$, the monopolistic price is $5$ and the second user loses. However, if she submits two separate transactions with a bid of $1$ each, the monopolistic price will be $1$ and all transactions will be included.
\label{ex:multibid}
\end{example}

Our empirical evaluation below accounts for such situations, showing that the benefit from using multiple bids also goes to zero as $n$ increases. For this purpose, we generalize $p^{S}$ to capture the possible benefit of splitting a single bid to $u$ different bids. Formally, define \nom{pmultibid}{$p^{multibid}(\bm b_{-i})$}{A generalization of $p^{S}$ to capture the possible benefit of splitting a single bid to $u$ different bids} as the minimum of $u \cdot b^{(u)}_i$ over all $u \in \N^{+}$ and all $b^{(1)}_i \geq \ldots \geq b^{(u)}_i \in \R$ such that $b^{(u)}_i \geq p^{M}(b^{(1)}_i,\ldots,b^{(u)}_i,\bm b_{-i})$. Choosing $u=1$ gives us $p^{S}$, hence $p^{S}(\bm b_{-i}) \geq p^{multibid}(\bm b_{-i})$. Example~\ref{ex:multibid} shows that there are cases in which the inequality is strict. It is easy to show that  $p^{multibid}(\bm b_{-i})$ can be written as:

\vspace{-6mm}

\begin{equation}
p^{multibid}(\bm b_{-i})\equiv \min \{u \cdot b_i\ |\ u \in \N^+,\ b_i\in \R,\ b_i\geq p^{M}(\overbrace{b_i,\ldots,b_i}^{\text{u times}},\bm b_{-i})\}.\label{eq:p_multibid}
\end{equation}

\noindent
The minimum is well defined using an argument similar to that used for $p^{S}$. In particular,
for any positive integer $u$, use $f_u(b_i) = p^{M}(b_i,...,b_i,b_{-i})$ (where $b_i$ appears $u$ times). By Claim~\ref{cl:u_cannot_be_large}, $u \in \{1,...,n\}$. Since for every $u$, the infimum is contained in the set, then so is the infimum over the union over $u \in \{1,...,n\}$.

We believe that Theorem~\ref{thm:delta_max_negligible} and Conjecture~\ref{con:delta_goes_to_zero} hold in the case of multiple bids, where in the definition of $\delta$, $p^{multibid}$ replaces $p^{S}$. In our empirical evaluation we use $p^{multibid}$ instead of $p^{S}$. In all the distributions we examined, except the discrete distribution, we {\em never} encountered a case in which the strategic player has an advantage placing multiple bids. Even in the discrete case, the effect of such multiple bids was negligible. 

The definition of $p^{multibid}$ gives little information on its algorithmic computation, since the optimization is done on an infinite set. We provide an alternative definition which yields a polynomial-time algorithm. We use this algorithm in our simulations. Fix a player $i$ and $\bm w = \bm v_{-i}$. Assume w.l.o.g.~that $\bm w$ is sorted, $w_1\geq w_2, \ldots, \geq w_{n-1}$. For every $k^*(\bm w) \leq j \leq n-1$, define
\[
f(j) \equiv \max\left\{ \left\lceil{\frac{R(\bm w)}{w_j} }\right \rceil ,j+1\right\}.
\]
Intuitively, if user $i$ (the strategic user) adds $f(j)-j$ additional bids to $\bm w$, all equal to $\frac{R(\bm w)}{f(j)}$, then this will be the new monopolistic price, and user $i$ will win. It is thus clear that these manipulations should be considered in order to determine the optimal solution of Eq.~\eqref{eq:p_multibid}. While a-priori it may be possible to consider other manipulations as well, the next theorem shows that considering this type of manipulations is sufficient.
\begin{theorem}
~~

\vspace{-4mm}

\[p^{multibid}(\bm w) = \min_{k^*(\bm w) \leq j \leq n-1} \frac{R(\bm w)}{f(j)}(f(j)-j).\] Furthermore, if $j^*$ is the index that minimizes this term, then $b^* = \frac{R(\bm w)}{f(j^*)}$ and $u^* = f(j^*) - j^*$ minimizes the r.h.s.~of Eq.~\eqref{eq:p_multibid}.
\label{thm:multibid_computation}
\end{theorem}

For intuition, we refer the reader to Example ~\ref{ex:multibid}. Indeed, in that example, $\bm w = (5,1,1), k^*(\bm w) = 1, R(\bm w) = 5$. We can calculate $f(1)=2,f(2)=5,f(3)=5$. Therefore, 
$j^*=3$, resulting in $b^*=1, u^*=2$ as described in the example above.

\begin{proof}
To prove this theorem, we begin with a useful Lemma:

\begin{lemma}
Let $b,u$ be the arguments which minimize Eq.~\eqref{eq:p_multibid}, $j$ be the integer which satisfies $w_j \geq b > w_{j+1}$ and $u^*=f(j)-j$. Then, $u^* \geq 1$ and $u \geq u^*$.
\label{lem:u_geq_u_start}
\end{lemma}
\begin{proof}
Since $f(j) \geq j+1$, $u^* \geq 1$. If $f(j) = j+1$, $u^* = 1$ and the claim immediately follows. Thus assume that $f(j) = \ceil{\frac{R(\bm w)}{w_j}}$. Suppose towards a contradiction that $u \leq u^*-1$. Then,
\begin{equation}
R(\overbrace{b,\ldots,b}^{\text{u times}},\bm w) =
b \cdot (u+j) \leq w_j \cdot (u^* -1 + j ) = w_j\cdot(f(j)-1) < R(\bm w)
\label{eq:revenue_in_p_multibid}
\end{equation}
where the first step follows since by Claim~\ref{cl:p_multibid-equality}
\[
b = p^{M}(\overbrace{b,\ldots,b}^{\text{u times}},\bm w)
\]
and since there are exactly $j$ bids in $\bm w$ that are at least $b$; and the last inequality follows from $f(j)<\frac{R(\bm w)}{w_j}+1$.

However Eq.~\eqref{eq:revenue_in_p_multibid} shows a contradiction since adding bids can only increase the monopolistic revenue. This completes the proof of Lemma~\ref{lem:u_geq_u_start}.
%The revenue from setting the monopolistic price to be $b$ is 
%$b\cdot(u+j)\leq w_j\cdot (u^* -1 +j ) = w_j(f(j)-1) < R(\bm w)$. 
\end{proof}

% \begin{thm}
% $p^{multibid}(\bm w) = \min_{k^*(\bm w) \leq j \leq n-1} \frac{R(\bm w)}{f(j)}(f(j)-j)$. Furthermore, if $j^*$ is the index that minimizes this term, then taking $b = \frac{R(\bm w)}{f(j^*)}$ and $u = f(j^*) - j^*$ minimizes the r.h.s.~in Eq.~\eqref{eq:p_multibid}.
% \end{thm}
% \begin{proof}
% We first show that $p^{multibid}(\bm w) \leq \min_{k^*(\bm v_{-i})\leq j\neq i \leq n-1} \frac{R(\bm v_{-i})}{f(j)}(f(j)-j)$. Take $ b,u$ as in the statement of the theorem. Claim~\ref{cl:b_u_j_revenue} shows that $u\in \N^+$ and 
%  \[
% b \geq p^{M}(\overbrace{b,\ldots,b}^{\text{u times}},\bm w)
% \]
% and since Eq.~\eqref{eq:p_multibid} takes the minimum over all $(b,u)$ pairs that satisfy this we have that $p^{multibid}(\bm w) \leq b \cdot u$.

% ???

% \end{proof}

%\begin{proof}
We now prove Theorem~\ref{thm:multibid_computation}. Let $b,u$ be the arguments that determine\\ $p^{multibid}(\bm w)$ (i.e., $(b,u)$ minimizes Eq.~\eqref{eq:p_multibid}). Let $j$ be the integer that satisfies $w_j \geq b > w_{j+1}$ (where $j=n-1$ if $w_{n-1} \geq b$). By Claim~\ref{cl:k_start_multibid}, $j \geq k^*(\bm w)$. Let $u^*=f(j)-j$ and $b^* = \frac{R(\bm w)}{f(j)}$. Then,

%\vspace{-6mm}

\[
b \cdot (u + j) = R(\overbrace{b,\ldots,b}^{\text{u times}},\bm w)
\geq R(\bm w) = b^* \cdot (u^*+ j)
\]
where the first step is the same as the first step in Eq.~\eqref{eq:revenue_in_p_multibid} above.
Thus, $b \cdot u \geq b^* \cdot u^* + j \cdot (b^* - b)$.
By Lemma~\ref{lem:u_geq_u_start}, $u \geq u^*$, which implies
$b \cdot u \geq b^* \cdot u^*$ (since $b \cdot u < b^* \cdot u^*$
and $u \geq u^*$ implies $b < b^*$ and $b \cdot u < b^* \cdot u^* + j \cdot (b^* - b)$, a contradiction). Thus,
\[
p^{multibid}(\bm w) = b \cdot u \geq b^* \cdot u^* \geq \min_{k^*(\bm w) \leq j \leq n-1} \frac{R(\bm w)}{f(j)}(f(j)-j),
\]
and the first part of the theorem follows. The second part follows in a straight-forward way from the first part.
\end{proof}

\section{Empirical Evaluation}
\label{sec:empirical_evaluation}
We provide empirical evidence to support the above conjectures and to supply greater detail on the rate at which the discount ratio converges to zero. We use both synthetic data as well as transaction data taken from the Bitcoin blockchain. The synthetic data are generated using four distributions:
\begin{enumerate}
\item A uniform discrete distribution over the integers $1,\ldots,100$. This distribution has a finite support size and therefore satisfies  the requirements of Theorem~\ref{thm:delta_max_negligible}. 
\item The uniform distribution over $[0,1]$. Notice that here the support size is infinite.
\item The half normal distribution.\footnote{Recall that the half normal distribution arises from taking the absolute value of a normal random variable. This is needed in our setting since valuations are non-negative.} Here, the probability for value $x$ decreases exponentially with $x$, hence, even though arbitrarily high transaction values will be seen, they are highly unlikely.
\item The \nom{inverse distribution}{inverse distribution}{$F(x)=1-\frac{1}{x},\quad x\in[1,\infty)$} 
\begin{equation}
 F(x)=1-\frac{1}{x},\quad x\in[1,\infty).
 \label{eq:inverse_distribution}
\end{equation}
Here, the probability for a value $x$ decreases polynomially with $x$.
\end{enumerate}

%\footnote{
\noindent
The uniform distribution has no tail, the half normal distribution has a light tail (it decreases exponentially fast), and the inverse distribution has a heavy tail.

%$\lim_{n\rightarrow \infty}\Delta^{\max}_n > 0$ for the inverse distribution as a result of its heavy tail. For any price $x\geq 1$, the revenue would be roughly $x \cdot (1-F(x))  = 1$, thus (up to noise) all prices yield the same revenue and a user can decrease her price by changing her bid. We do not think that this is the actual distribution of maximum transaction fees but include it here to demonstrate bad cases.
%}
\begin{figure}
\centering \includegraphics[width=0.8\columnwidth]{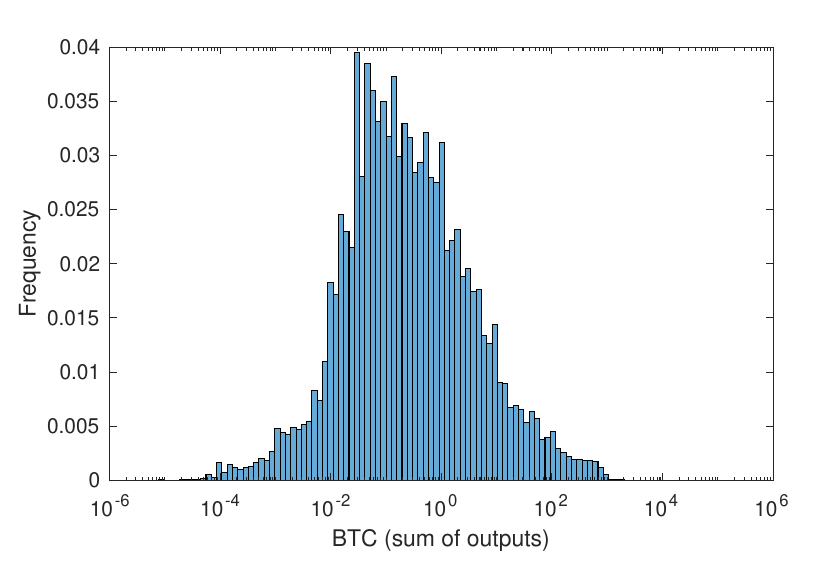}
\caption{A histogram of the sum-of-outputs on the x-axis (in log scale) and  its frequency in the collected data (on the y-axis).}
\label{fig:historgram_sum_of_outputs}
\end{figure} 

Bitcoin blockchain data was collected from 1000 consecutive blocks, which constitute roughly one week of activity ending on October 28th, 2016.
Figure~\ref{fig:historgram_sum_of_outputs} shows a histogram of the collected data.
The data obviously do not contain the values, since this is not how Bitcoin operates. The simulations estimate the value $v$, as a function $v=v(x)$ of the transaction size $x$ ($x$ is the sum of all outputs of the Bitcoin transaction). We use three alternative functions: $v(x) = \log x, v(x) = \sqrt{x}, v(x)=x$. Note that multiplying all bids by a scalar does not change the discount ratio, which is the case when, e.g., the value is some percentage of $x$.

We empirically evaluated $\Delta_n^{average}$ and $\Delta_n^{max}$ (where $p^{multibid}$ replaces $p^{S}$) as a function of the number of bids, $n$. For each $n \in \{2^3,2^4,\ldots,2^{17}\}$, we conducted 100 simulation runs. Each run samples $n$ bids and calculates the $n$ empirical discount ratios. $\Delta_n^{average}$ is the average of the $n$ individual discount ratios and $\Delta_n^{max}$ is the maximum over the $n$ individual discount ratios. Each point in the graph is the average of 100 runs.

\begin{figure}
\centering \includegraphics[width=\columnwidth]{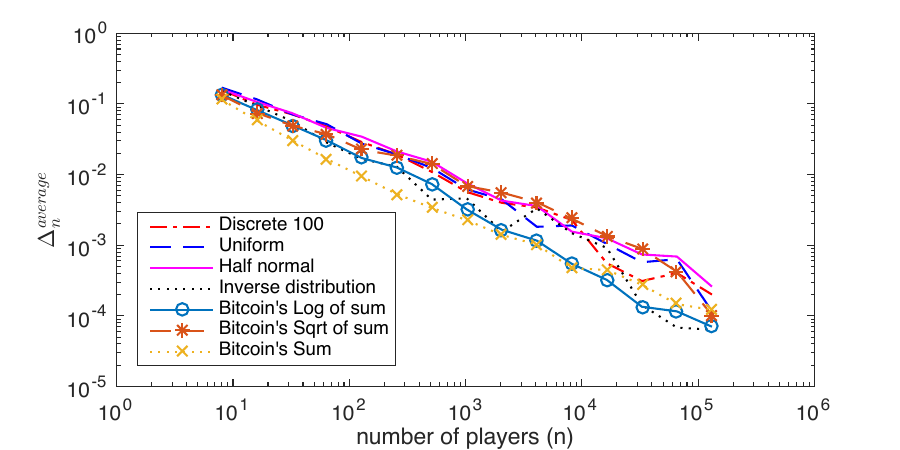}
\caption{The average discount ratio of a player from selfish bidding as a function of the number of players that participate.}
\label{fig:delta_avg}
\end{figure}

Figure~\ref{fig:delta_avg} shows that $\Delta_n^{average}$ behaves almost identically for all tested distributions. In particular, as stated in the first part of Conjecture~\ref{con:delta_goes_to_zero}, $\Delta_n^{average}$ decreases linearly with the number of bids $n$ and for all distributions, even those with an infinite and unbounded support size.

Figure~\ref{fig:delta_max} shows that $\Delta_n^{max}$ behaves differently for some of the tested distributions. For the uniform distribution, $\Delta_n^{max}$ decreases linearly with the number of bids, supporting the second part of Conjecture~\ref{con:delta_goes_to_zero}. The half normal distribution behaves similarly (even though it is not bounded). For the inverse distribution, $\Delta_n^{max}$ does not seem to decrease with the number of users, and we believe that this is in fact an example that supports the third part of Conjecture~\ref{con:delta_goes_to_zero}. The three Bitcoin distributions behave as follows: the log of the sum decreases the fastest, the square root of the sum decreases more slowly, and the sum decreases the slowest. For example, for $n=2048$ (roughly the current Bitcoin block-size), $\Delta_n^{max}$ is about $0.2\%, 5\%,$ and $19\%$ for the three distributions, respectively. When $n=2^{17} \approx 130,000$, $\Delta_n^{max}$ is about $0.0007\%, 0.05\%,$ and $1\%$.

\begin{figure}
\centering \includegraphics[width=\columnwidth]{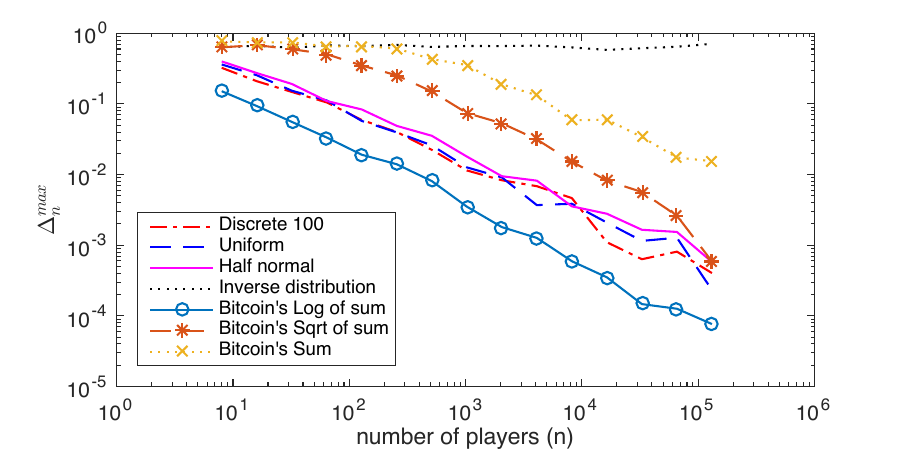}
\caption{The maximal discount ratio of a player from selfish bidding as a function of the number of players.}
\label{fig:delta_max}
\end{figure}

These findings show that there is a qualitative difference between $\Delta_n^{average}$ and $\Delta_n^{max}$, mainly for the inverse distribution and for the Bitcoin distribution with $v(x)=x$. In these distributions, although the average transaction will benefit very little from bid shading, some transactions can nonetheless obtain significant benefit. This outcome can, in principle, result from two different reasons. The first is that $\Delta_n^{average}$ also accounts  for the losing transactions, whose discount ratio is zero. The second is that $\Delta_n^{max}$ only accounts for the single winning user with the maximal discount ratio. However, typical user behavior (as demonstrated in Figure~\ref{fig:p_strategic}) shows that almost all winning bids have the lowest strategic price and, therefore, the highest discount ratio. Thus, the explanation  for the difference between the two discount ratios seems to be the first.

Figure~\ref{fig:log_bitcoin_scatter_plot} shows all simulation points for the Bitcoin log-of-sum distribution. The x-axis is the resulting {\em block size} of the simulation run (in the previous figures it was $n$) and the y-axis is the maximal discount ratio of the simulation run. Three main conclusions can be drawn from Figure~\ref{fig:log_bitcoin_scatter_plot}. First, block sizes (number of winners) range from 1 to about 25,000 when $v(x)= \log x$ (recall that the total number of bids ranges from $2^3$ to $2^{17}$).\footnote{The throughput when $v(x)= \log x$ will therefore be about $0.1n$ ($n$ is the number of all transactions). For $v(x)= x$, block sizes range from 1 to about 1,500, implying a throughput of about $0.01n$. This can be compared to Bitcoin's actual throughput in October 2016 which varied between 1300 and 1900 transactions per-block (the variation is mostly due to variations in transaction sizes in bytes).} Second, this range in block sizes is the main reason why $\Delta_n^{max}$ is smaller in the Bitcoin log-of-sum distribution, as larger block sizes imply smaller discount factors. Third, the distribution of the simulation points, which include outliers, is not normal. We do not have an explanation for this finding. The last two remarks apply to all tested distributions.

\begin{figure}
\centering \includegraphics[width=\columnwidth]{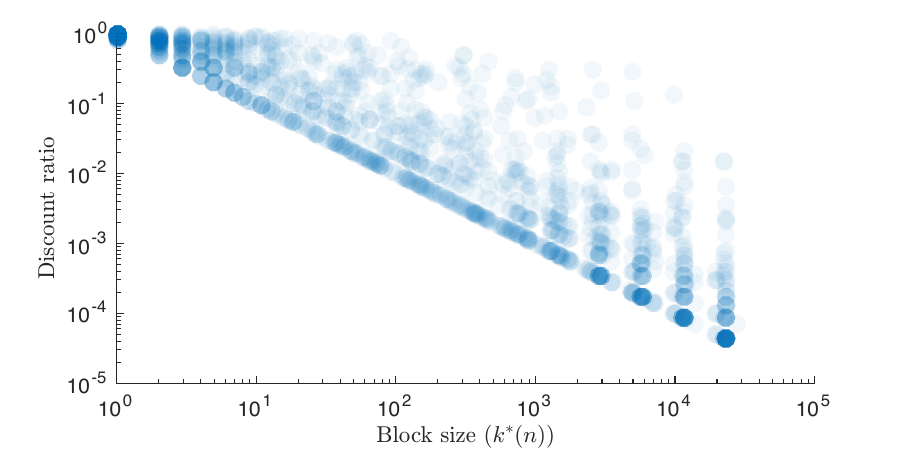}
\caption{A scatter plot of all simulation points for the Bitcoin distribution with $v(x) = \log x$.}
\label{fig:log_bitcoin_scatter_plot}
\end{figure}

\section{The Suitability of the RSOP Auction in Our Context}
\label{sec:RSOP}

The monopolistic-price auction is miner-honest in the sense that a myopic miner cannot increase her revenue by diverging from the suggested protocol (for example, add false bids or delete existing bids). However, as discussed above, even impatient users can sometimes benefit from submitting a strategic bid which is different than their true value.

In this section we discuss an alternative approach in which we use a well-known auction which is truthful for impatient users. That is, an impatient user will maximize her gain by simply submitting her value. On the other hand, in this second mechanism the challenge is to show that miners will be honest and will not manipulate the protocol, as we further discuss below. In particular, we will show empirical evidence suggesting that the profit for the miners from such manipulations becomes negligible as the number of users increases.

More specifically, our second suggested mechanism uses the following RSOP (Random Sampling Optimal Price) auction, defined by~\cite{goldberg2006competitive}. In the definition, the following notation is used: For a subset $A\subseteq [n]$ of users, and a bid vector $\bm b$, let $\bm b_A=(b_i)_{i \in A}$. We define $p^{M}_A(\bm b)\equiv p^{M}(\bm b_A)$. When $\bm b$ is clear from the context, we simply write $p^{M}_A$.

\begin{definition}[The \nom{RSOP}{RSOP}{The Random Sampling Optimal Price Auction, Definition~\ref{def:RSOP_auction}} auction~\cite{goldberg2006competitive}]
Upon receiving $n$ bids, $\bm b = b_1,....,b_n$, the miner performs the following:
\begin{enumerate}
\item Randomly partition the bids to two disjoint sets $A$ and $B$ (each bid is placed in $A$ with probability $\frac{1}{2}$, otherwise it is placed in $B$).
\item Compute the monopolistic price for each set: $p_A^{M}$, $p_B^{M}$. The monopolistic price of an empty set is defined to be zero.
\item The set of winning bids is $A' \cup B'$, where:
$$A' = \{i\in A : b_i\geq p^{M}_B\} \quad ; \quad B' = \{i\in B : b_i\geq p^{M}_A\}.$$ The transactions in $A'$ each pay $p^{M}_B$, and the transactions in $B'$ pay $p^{M}_A$ each.
The revenue obtained in the auction is therefore $$RSOP(\bm b)=|A'|\cdot p^{M}_B+ |B'|\cdot p^{M}_A.$$
\end{enumerate}
\label{def:RSOP_auction}
\end{definition}

%Apply the monopolistic price from set A to the transactions in set B, and vice versa, i.e., the set of accepted transactions from each set is the set that agreed to pay more than the monopolistic price of the other set: $$A' = \{i\in A : b_i\geq p^{M}_B\} \quad ; \quad B' = \{i\in B : b_i\geq p^{M}_A\},$$ and the revenue of the mechanism is $$RSOP(\bm b)=|A'|\cdot p^{M}_B+ |B'|\cdot p^{M}_A.$$

\noindent
For the RSOP mechanism, Goldberg \etal~provide the following results that hold 
under ``the usual'' auction theory assumptions, most notably that: (1) the auctioneer is honest; (2) bidders do not collude; (3) each bidder can submit exactly one bid; (4) bidders are impatient, i.e. they do not obtain utility from winning in future auctions (we discuss these assumptions further below).

\begin{theorem}[\cite{goldberg2006competitive}, Observation 6.2 and Theorem 6.4]
~~
\begin{enumerate}
\item \emph{Truthfulness.} The RSOP auction is truthful, i.e., a bidder maximizes her utility by reporting her value, even when the other bidders do not reveal their value. \\

\item \emph{Maximal Revenue.} Fix any parameter $h$. Let $\bm b$ be any bid vector of $n$ bids with $b_i \in [1,h]$ for all $i$. Then $$\lim_{n\to \infty} \max_{\bm b} \frac{R(\bm b)}{RSOP(\bm b)}=1.$$
%$$\lim_{n\to \infty} \max_{\bm b = (b_1,...,b_n) \mbox{ s.t. } \forall i,~b_i \in [1,h]} \frac{R(\bm b)}{RSOP(\bm b)}=1.$$
\end{enumerate}
  \label{thm:RSOP}
\end{theorem}

To instantiate the RSOP-based mechanism, users must specify a bid that ideally represents their values for each transaction, and the miners are asked to create a block with all transactions they wish to \emph{potentially} include. Unlike the current Bitcoin protocol, here not all transactions in a block are valid. After the block is mined, and propagated to the Bitcoin nodes in the network, they determine which transactions are valid, by running 
Algorithm~\ref{alg:RSOP_block_verification}.

% \begin{algorithm}[H]
% 		\caption{RSOP mining mechanism}\label{alg:genCompr}
% 		\begin{algorithmic}[1]
% 			\STATE The miner uses the current mining mechanism, but includes all the transaction she is aware of in the block.
%             \STATE If the block 
% 		\end{algorithmic}
%         \label{alg:meta_algorithm}
% \end{algorithm}

\begin{algorithm}[H]
		\caption{RSOP block verification}\label{alg:genCompr}
		\begin{algorithmic}[1]
			\STATE A node receives a new block $B$.
            \STATE Check validity of the block and of all the included transactions according to Bitcoin's current rules. When referring to transactions in previous blocks, consider only valid transactions (as explained below).
%             \STATE Check that all the transaction in the block are legal (i.e. the digital signatures are valid, the sum of the inputs is larger than the sum of the outputs, etc.).
            \STATE \label{line:PRNG} Compute the sets $A$ and $B$ using the block hash as a seed to a Cryptographically Secure Pseudo-Random Number Generator (CSPRNG). 
            \STATE The transactions that are considered valid are the ones in $A'$ and $B'$ as in Def.~\ref{def:RSOP_auction}.
            \STATE
            \label{step:alpha}
            Transactions in $A'$ pay $p_B^{M}$ and transactions in $B'$ pay $p_A^{M}$. A fraction $1-\alpha$ of that revenue goes to the miner who mined the current block, and the rest goes to the future miner who would mine the next valid block. Invalid transactions do not pay anything. The parameter $0\leq \alpha \leq 1$ needs to be specified as part of the protocol. 
		\end{algorithmic}
        \label{alg:RSOP_block_verification}
\end{algorithm}

%To instantiate the RSOP mechanism, the miners are asked to create a block with all transactions they wish to potentially include. Unlike the current Bitcoin protocol, here not all transaction in a block are valid. After the block is mined, and propagated to the Bitcoin nodes in the network, they determine which transactions are valid, by running the RSOP auction locally in the following way: they partition the set of transactions into the sets A,B using a PRNG seeded with the block hash, and declare the transaction which win in the RSOP auction as valid. Since the source of the randomness (block hash) is the same for all the Bitcoin nodes, all nodes will agree on the set of valid transactions. 

% It is important to notice, that according to the RSOP mechanism, not all transactions agree to pay the price that is set for them, implying that not all transactions that are written in the block are eventually considered approved. Since the miner is generating a proof-of-work for the block, he does not know the block hash that will be produced until after he successfully creates the block. He thus cannot predict the partition into sets. 
% \anote{improve the explanation above. This is too hand wavy and insecure.}

The algorithm has the following advantages.  Since the source of the randomness (block hash) used in line~\ref{line:PRNG} is the same for all the Bitcoin nodes, all nodes will reach consensus on the set of valid transactions. Notice that the block hash is determined by its contents including the solution to a proof-of-work puzzle. A miner cannot easily manipulate the choice of the sets $A$ and $B$ to maximize revenue\footnote{A more careful cryptographic analysis is required to establish this statement, but is outside the scope of this work. This statement is fairly simple to show in the random oracle model.} (if she could, it could destroy the truthfulness of the protocol). She \emph{can} choose to forgo a block in which $A$ and $B$ were not set to her liking but will need to generate the proof-of-work from scratch, which is extremely costly.  There are however several concerns with the proposed algorithm:

\begin{trivlist}

\item \emph{User truthfulness.}
 The assumption that users cannot send multiple bids is unrealistic in the Bitcoin setting. Therefore, we do not know whether the suggested protocol is truthful even for impatient users and honest miners. 
 Under the assumption that the user cannot control which bids will be associated with subset $A$ and which with subset $B$, we do not know how to construct even tailored worst-case examples or if there exists an efficient algorithm to find beneficial deviations.\footnote{For the monopolistic auction, we discuss this issue in Sections~\ref{sec:p_multibid} and ~\ref{sec:empirical_evaluation}.}

\vspace{3mm}

\item \emph{Miner honesty: adding false bids.}
Miners are able to manipulate the protocol by adding false bids, especially if at step~\ref{step:alpha} of Algorithm~\ref{alg:genCompr}, $\alpha=0$. The purpose of choosing $\alpha>0$ is to mitigate this problem. The following example demonstrates this issue:
\begin{example}
Suppose $\alpha=0$. There are two users with $b_1 = h,\ b_2 = \ell$ where $h>2\ell$. With probability $\frac{1}{2}$ both users will fall to the same set, and the revenue for the miner would be $0$. If they fall into different sets, only the $h$ user will be included, and will pay $\ell$. Therefore, the expected revenue is $\frac{\ell}{2}$. A strategic miner can create many false transactions that bid $h$, by this receiving a revenue of approximately $h$. The false bids clearly cannot create a profit for the miner but also do not harm her because she pays them to herself. As $\alpha$ increases, the profitability of this strategy decreases.
\end{example}

More generally, a dishonest miner can always use the following simple strategy to obtain the revenue $R(b)$: (1) Compute the monopolistic price over all bids; (2) Add sufficiently many false transactions whose bids are equal to the monopolistic price. Adding sufficiently many such false bids will change the RSOP revenue to that of the monopolistic price mechanism.

We conjecture that the strategy in the preceding paragraph is always beneficial to the miner, when $\alpha=0$:
\begin{conjecture}
\label{conjecture:RSOP_revenue_smaller_than_monopolistic_revenue}
For every $\bm b$ and all choices of $A$ and $B$, the RSOP revenue is at most the monopolistic revenue. In particular, $RSOP(\bm b)\leq R(\bm b)$.
\end{conjecture}
The right hand side $(R(\bm b))$ is the revenue a manipulating miner can obtain with the strategy above, and the left hand side is the RSOP revenue, for \emph{any} allocation of the sets A and B (even one chosen adversarially). When $\alpha=0$, this is the revenue of the miner which uses Algorithm~\ref{alg:RSOP_block_verification} honestly, i.e., without false bids.

In the empirical evaluation reported below, we have never encountered a counterexample to Conjecture~\ref{conjecture:RSOP_revenue_smaller_than_monopolistic_revenue}. 
%\footnote{We do not know if there are cases in which these additional false bids will decrease the revenue. However, an empirical evaluation that we have conducted did not manage to obtain even one such case.} 
To measure the expected ``gain ratio'' from performing this strategy, we define
\[
\Delta^{RSOP}_n \equiv \E_{(v_1,\ldots,v_n)\sim F}\left[\frac{R(\bm v)}{RSOP(\bm v)}-1 \right].\]

\noindent
For distributions with bounded support sizes
Theorem~\ref{thm:RSOP} proves that \nom{deltarsop}{$\Delta^{RSOP}_n$}{$\Delta^{RSOP}_n \equiv \E_{(v_1,\ldots,v_n)\sim F}\left[\frac{R(\bm v)}{RSOP(\bm v)}-1 \right]$} goes to zero as $n$ goes to infinity. Furthermore, Figure~\ref{fig:Delta_RSOP} demonstrates that this indeed happens in a reasonably fast manner for almost all the distributions we considered in Section~\ref{sec:empirical_evaluation}, with the only exception being the inverse distribution.
\begin{figure}
\includegraphics[width=\textwidth]{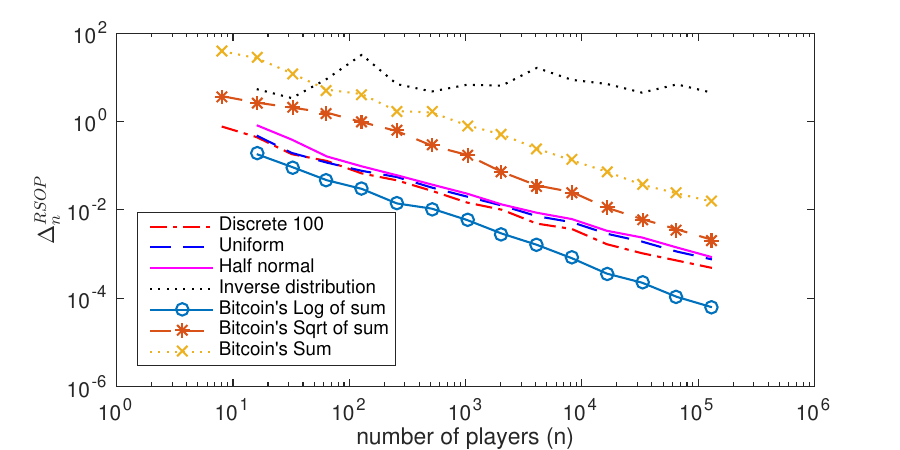}
\caption{$\Delta^{RSOP}_n$ as a function of the number of players, for various distributions.}
\label{fig:Delta_RSOP}
\end{figure}
To summarize this issue:

\begin{itemize}

\item {\bf Gaining from adding false bids when $\alpha=0$. } In the RSOP mechanism the miner can gain from adding false transactions with a bid equal to the monopolistic price. This is an easy strategy to perform, and as far as we can tell, it never harms the miner when $\alpha=0$. If the miner does so, the mechanism changes and becomes similar to the monopolistic-price mechanism. This, in turn, may harm the bidder-truthfulness property. 
However, the gain ratio from performing this strategy decreases as the number of users increases (with the exception of the Inverse distribution) -- see Figure ~\ref{fig:Delta_RSOP}.

\item {\bf Choosing $\alpha>0$ to eliminate the gain from adding false bids. } We conjecture that, for many bid distributions and when the number of users is fairly large, it is possible to set $\alpha > 0$ in a way that will eliminate the profitability of this strategy. Note that if $\alpha > 0$, the miner needs to pay some of her false bids to other miners.
%\item
On the other hand, setting $\alpha>0$, and especially high values (say, $\alpha=0.9$), introduces risks of side-payments. A miner may provide the following service. She asks users to submit transactions with 0 bid, and pay directly to the miner a fixed per transaction fee. The miner will guarantee that she will eventually include these transactions in a block which contain only other 0 bid transactions (which will then all be valid). In such a case the miner gets all the revenue for herself, whereas if she follows the protocol she will only receive a fraction $1-\alpha$ of the payments. If $\alpha$ is low, this is less likely to happen, as this requires non-trivial coordination between the miner and the user (e.g., trust regarding the transfer of payments, regarding inclusion in the block, etc).
\end{itemize}

\item \emph{Miner honesty: removing true bids.} The miner might also be able to increase her profit by removing true bids. The following example demonstrates this:
\begin{example}
\label{ex:2n1n}
There are $n$ bids with value $2$ and $n$ bids with value $1$. The miner can obtain a revenue of $2n$ w.h.p. by removing all bids of value $1$. The expected revenue for the honest miner is always at most $2n$, and in the limit of large $n$, the revenue is $2n-\sqrt{ \frac{n}{\pi} }+ o(\sqrt{n})$, see Appendix~\ref{sec:analysis_RSOP_example} Cor.~\ref{cor:RSOP_example}.
\end{example}
Note that the maximal revenue that can be obtained in this example by removing bids is exactly the monopolistic revenue. We do not know if this is true in general, however we strongly suspect that this is the case. More specifically, we conjecture:

\begin{conjecture}
For any vector of bids $\bm b$,
$\max_{\bm b' \subseteq \bm b} RSOP(\bm b') \leq R(\bm b)$.
\label{con:RSOP_remove_bids_smaller_revenue}
\end{conjecture}
In fact, Conjectures~\ref{conjecture:RSOP_revenue_smaller_than_monopolistic_revenue} and~\ref{con:RSOP_remove_bids_smaller_revenue} are equivalent: by choosing $\bm b'=\bm b$, Conjecture~\ref{con:RSOP_remove_bids_smaller_revenue} implies ~\ref{conjecture:RSOP_revenue_smaller_than_monopolistic_revenue}. To show the other direction, Conjecture~\ref{conjecture:RSOP_revenue_smaller_than_monopolistic_revenue} implies that $RSOP(\bm b')\leq R(\bm b')$ and by the monotonicity of $R(\cdot)$, we reach $RSOP(\bm b')\leq R(\bm b')\leq R(\bm b)$, as required.  

Note that under this conjecture, the gain ratio $\Delta^{RSOP}_n$ goes to zero for any distribution with finite support even if the miner can remove true bids, using Theorem~\ref{thm:RSOP}. Note that this form of the conjecture is harder to verify empirically as it requires considering different subsets for removal. But the equivalence to Conjecture~\ref{conjecture:RSOP_revenue_smaller_than_monopolistic_revenue}, and the empirical evaluation performed above provides some supporting evidence for it. 

We do not know an efficient algorithm to find the best bids to remove, and therefore the strategy for the miner is left open. We do not have any suggestions regarding how to mitigate this case. 
Unlike the previous case (adding false bids) $\alpha$ is irrelevant to this strategy by the miner. Perhaps it can be shown that removing true bids is rarely beneficial, under some distributional assumptions.

\item \emph{Block size.}
Notice that in the RSOP mechanisms blocks may contain many transactions that do not get accepted (their bid may be below the price that is eventually determined). One possible way to prevent this is to commit to all transactions in the Merkle tree in some canonical order of bid size and then eventually only reveal transactions that are needed to establish the monopolistic price and transactions that win this bid. While we do not have a fully fleshed out scheme, we do believe that clever use of data structures and cryptographic schemes may help in reducing the amount of wasted space in blocks in this way. 
\end{trivlist}

\section{Discussion}
\label{sec:summary}

The simple pay-your-bid fee mechanism used in Bitcoin has several disadvantages:
%Its most notable advantage is the resilience to miner manipulations (aimed to increase revenue). However, there are several disadvantages:
(i) A fixed (hard-coded) maximal block size implies non-optimal revenue extraction when the pre-fixed block-size is too small or too large for the current bid instance, (ii) Bitcoin wallets must invest computational effort in bid shading, and (iii) deciding on the maximal block size involves economic rather than purely technological considerations  (e.g., block propagation time~\cite{decker13information,Sompolinsky15Secure}). For single block creation, the proposed monopolistic auction effectively solves all these issues while maintaining miner non-manipulability.

%{\bf The difference from standard game-theoretic analysis. } In a Bayesian-Nash equilibrium (BNE) analysis, user $i$'s utility is $u_i(v_i, b_i, \bm b_{-i}) = v_i - p^{M}(\bm b)$ if $b_i \geq p^{M}(\bm b)$ and $0$ otherwise. A strategy $s_i: \R \rightarrow \R$ gives $i$'s bid $s_i(v_i)$. $(s_1(\cdot),...,s_n(\cdot))$ is an $\epsilon$-relative BNE if $\forall i,v_i,b_i: \E_{\bm v_{-i}}  [u_i(v_i, s_i(v_i), \bm v_{-i})] \geq (1-\epsilon)\E_{\bm v_{-i}}  [u_i(v_i, b_i, \bm v_{-i})]$.Our discount-factor analysis is very similar, but formally different in three main aspects. First, $\delta_i(\bm v) \leq \epsilon$ iff $p^{S}(\bm v_{-i}) \geq (1 - \epsilon)p^{M}(\bm v)$, while, taking the best $b_i$, $u_i(v_i, s_i(v_i), \bm v_{-i})] \geq (1-\epsilon)u_i(v_i, b_i, \bm v_{-i})$ iff $p^{S}(\bm v_{-i}) \geq p^{M}(\bm v)/(1 - \epsilon) - \epsilon \cdot v_i$. I.e., as $v_i$ gets larger, a larger price gain from manipulation will still be considered a small utility gain, while the discount factor considers only the price gain. Second, the deviation $b_i$ in the BNE analysis is fixed before taking the expectation, while in our analysis $b_i$ is a best response to a specific $v_{-i}$ and we take the expectation over these possibly different best responses. In this respect our analysis is clearly stronger. Third, the BNE analysis considers the ratio of expectations, while we consider the expected ratio (the expected discount factor).

%\vspace{1mm}

%\noindent
%{\bf The RSOP auction. }

The monopolistic auction is $\epsilon$-truthful ($\epsilon$ goes to zero as the number of users grows). Alternatively, we can take a truthful auction where a user cannot influence its winning price, e.g., the RSOP auction is truthful and its revenue $RSOP(\bm v)$ satisfies $\lim_{n\to \infty} \max_{\bm v} \frac{R(\bm v)}{RSOP(\bm v)}=1$~\cite{goldberg2006competitive}. Section~\ref{sec:RSOP} describes an implementation in Bitcoin (e.g., the required randomness needs to be verified by other miners) and discusses major drawbacks of it: First, to verify correctness, the block must contain all rejected transactions, creating significantly larger blocks. Second, truthfulness does not rule out beneficial bid splits. We do not know whether RSOP is resilient to such a manipulation. We discussed the effect of bid splits in the monopolistic auction in Section~\ref{sec:p_multibid}. Third, although
$\frac{R(\bm v)}{RSOP(\bm v)} \approx 1$, nonetheless $R(\bm v) - RSOP(\bm v)$ still usually grows to infinity (e.g., if $R(\bm v) = O(n + \sqrt{n})$ and $RSOP(\bm v) = O(n)$). Thus, a miner can significantly increase its revenue by engaging in manipulative behavior, for example, (i) deletion of true bids (however, we do not know if there exists a polynomial-time algorithm to find which bids to delete), (ii) adding many false bids, all of which are equal to the monopolistic price of the true bids (a possible solution to this attack is to split the block fees between the miner that created the block and the miner that will create the next block). 
In light of these shortcomings,  we believe that the monopolistic auction is better suited to implement improvements in Bitcoin's fee market.

\subsection{Temporal considerations}
\label{sec:temporal_considerations}
An important issue we did not cover is how the current mechanism would behave in a realistic setting where the interaction with miners is repeated. 
Just as in other mechanisms (for example, the second price auction) moving from a single shot interaction to a repeated one does not necessarily maintain the mechanism's truthfulness. 
In our case, one could consider a repeated setting, in which users who did not manage to have their transactions included in a block still persist and request a transfer from the miners again (indeed, by default, Bitcoin transactions persist in the miner's queue if they were not already included in a block). In this setting, the monopolistic mechanism, which clears the queue of its highest value transactions will slowly cause the distribution of transactions in the queue to be more skewed towards lower-bid transactions compared to the distribution of new incoming transactions. We then expect that occasionally the mechanism will cause a price fluctuation, when the weight of lower-bid transactions shifts the monopolistic price down. In this case, the occasional drop in price may entice patient users (who we did not consider in our model) to bid lower values and have their transactions accepted whenever such a drop occurs. The mechanism is therefore not truthful if applied in this manner. 

Another important aspect to consider if patient users are present is the effect on the strategic behavior of miners. We have argued that in the one-shot monopolistic auction, miners have no incentive to add false bids or delete true ones. We believe this will be the case for small miners even when patient users exist, since small miners cannot anticipate to mine blocks often. However, with large miners (or mining pools), a different analysis will be required. Consider an example where at every round, 51\% or new transactions are of value 1, and 49\% of transactions are of value 2. In this case, a myopic miner may include all transactions in his block with a monopolistic price of 1. A more far-sighted miner that has almost 100\% of the hashing power\footnote{Of course, such a miner could double-spend transactions, censor users, etc., so, in some sense, the manipulation which we discuss is the least of our worries.} may instead alternate between blocks that contain only transactions with value 1 and blocks that only contain transactions with value 2, relying on the ability to defer transactions by one block, and gaining in the process.

% \onote{Are there any non-trivial alternative auctions with better properties?}

On the face of it, a simple posted-price technique could handle such temporal considerations. Specifically, gather statistics about the underlying distribution and calculate an optimal fixed price for transactions to be included in the block. However, since the true underlying distribution is really unknown, such a method is not straight-forward and contains several significant challenges. First, it is hard to learn the full distribution, since blocks only contain information about accepted transactions, and therefore there is no consensus regarding the tail of the bidding distribution. Second, the distribution itself may change as a result of the mechanism's own behavior if transactions that are not included in blocks persist and attempt to enter future blocks. Third, such a statistical method is highly sensitive to miner manipulations as they can easily alter sampled transactions and their transaction fees. One recent quite sophisticated attempt to handle these issues is contained in Ethereum's recent proposal EIP1559~\cite{buterin2019eip} which is examined thoroughly in~\cite{rougharden2020transaction} (which also assumes that users are myopic with regards to their transactions, as we do) and in a related analysis in~\cite{ferreira2021dynamic}. A key idea used within EIP1559 proposes burning money to avoid the miners manipulating the posted amount that transactions are charged. The analysis in \cite{ferreira2021dynamic} further shows that EIP1559 is not always stable and a different posted price mechanism is proposed. All this demonstrates that successfully handling such issues even in a myopic setting and we believe that this is an open problem which deserves and requires significant further effort. 

\subsection{Future Work}
The monopolistic auction can also be explored further via a careful equilibrium analysis. The astute reader may notice that the current model allows for an equilibrium in which all users bid 0. This is an unrealistic setting since in practice there exists a maximum limit on the block size, and not all transactions can be included in the block. 
%We believe that in a setting in which not all transactions can be included in the block, no other meaningful equilibria exist other than truthfulness, which is an $\epsilon$-BNE. To reach a definitive conclusion, however, a careful analysis is required.
Still, an elaborate game-theoretic analysis of all possible equilibria in our model is interesting.

One of the modeling uncertainties in this work regards the users' values. We chose to analyze various distributions (see, e.g., Figure~\ref{fig:delta_max}) since we have little information regarding the maximal willingness to pay of the users. This issue has been studied extensively in several economic contexts~\cite{nussair04revealing,lange2002impact,wicker2011willingness} and we believe that studying it in Bitcoin will be beneficial.

The work in~\cite{tsabary18gap} shows that, given Bitcoin's current reward scheme, it is beneficial to miners to introduce a time gap between the time a miner observes the last block and the time the miner starts mining the new block. It could be interesting to examine whether our proposed new reward scheme eliminates this issue as well.

\section{Acknowledgements} 
A preliminary extended abstract of this paper appeared in~\cite{lavi19redesigning}.
We thank Alaa Mozalbat, who designed and implemented the simulation code, and Ittay Eyal and Reshef Meir for valuable discussions. R.L.~was partly supported by a Marie-Curie fellowship ``Advance-AGT'', by an ARCHES award from the MINERVA foundation, and by the ISF-NSFC joint research program (grant No.~2560/17). O.S.~was supported by ERC Grant 280157, by the Israel Science Foundation (personal grant 682/18 and Quantum Technologies and Science Program grant 2137/19), and by the Cyber Security Research Center at Ben-Gurion University. 
 A.Z.~was supported by the Israel Science Foundation (grant 616/13) and by a grant from the HUJI Cyber Security Research Center in conjunction with the Israel National Cyber Bureau (grant  039-9230).

\bibliographystyle{plain}
\bibliography{An_auction_approach_to_the_Bitcoin_block_size_challenge}

\begin{thebibliography}{10}

\bibitem{abe2002m+}
Masayuki Abe and Koutarou Suzuki.
\newblock M+1-st price auction using homomorphic encryption.
\newblock In David Naccache and Pascal Paillier, editors, {\em Public Key
  Cryptography, 5th International Workshop on Practice and Theory in Public Key
  Cryptosystems, {PKC} 2002, Paris, France, February 12-14, 2002, Proceedings},
  volume 2274 of {\em Lecture Notes in Computer Science}, pages 115--124.
  Springer, 2002.

\bibitem{azevedo18strategy}
Eduardo~M Azevedo and Eric Budish.
\newblock {Strategy-proofness in the Large}.
\newblock {\em The Review of Economic Studies}, 86(1):81--116, 08 2018.

\bibitem{azouvi19tools}
Sarah Azouvi and Alexander Hicks.
\newblock Sok: Tools for game theoretic models of security for
  cryptocurrencies, 2019.

\bibitem{babaioff12bitcoin}
Moshe Babaioff, Shahar Dobzinski, Sigal Oren, and Aviv Zohar.
\newblock On {B}itcoin and red balloons.
\newblock In {\em {ACM} Conference on Electronic Commerce, {EC} '12, Valencia,
  Spain, June 4-8, 2012}, pages 56--73. {ACM}, 2012.

\bibitem{basu19towards}
Soumya Basu, David Easley, Maureen O'Hara, and Emin~G{\"{u}}n Sirer.
\newblock Towards a functional fee market for cryptocurrencies, 2019.

\bibitem{bogetoft2006practical}
Peter Bogetoft, Ivan Damg{\aa}rd, Thomas~P. Jakobsen, Kurt Nielsen, Jakob
  Pagter, and Tomas Toft.
\newblock A practical implementation of secure auctions based on multiparty
  integer computation.
\newblock In Giovanni~Di Crescenzo and Aviel~D. Rubin, editors, {\em Financial
  Cryptography and Data Security, 10th International Conference, {FC} 2006,
  Anguilla, British West Indies, February 27-March 2, 2006, Revised Selected
  Papers}, volume 4107 of {\em Lecture Notes in Computer Science}, pages
  142--147. Springer, 2006.

\bibitem{bonneau16why}
Joseph Bonneau.
\newblock Why buy when you can rent? - bribery attacks on bitcoin-style
  consensus.
\newblock In {\em Financial Cryptography and Data Security - {FC} 2016
  International Workshops, BITCOIN, VOTING, and WAHC, Christ Church, Barbados,
  February 26, 2016, Revised Selected Papers}, pages 19--26. Springer, 2016.

\bibitem{buterin2019eip}
Vitalik Buterin, Eric Conner, Rick Dudley, Matthew Slipper, Ian Norden, and
  Abdelhamid Bakhta.
\newblock {EIP}-1559: Fee market change for {ETH} 1.0 chain, 2019.

\bibitem{bohme15bitcoin}
Rainer Böhme, Nicolas Christin, Benjamin Edelman, and Tyler Moore.
\newblock Bitcoin: Economics, technology, and governance.
\newblock {\em Journal of Economic Perspectives}, 29(2):213--38, May 2015.

\bibitem{carlsten16instability}
Miles Carlsten, Harry~A. Kalodner, S.~Matthew Weinberg, and Arvind Narayanan.
\newblock On the instability of {B}itcoin without the block reward.
\newblock In {\em Proceedings of the 2016 {ACM} {SIGSAC} Conference on Computer
  and Communications Security, Vienna, Austria, October 24-28, 2016}, pages
  154--167, 2016.

\bibitem{chen2019axiomatic}
Xi~Chen, Christos~H. Papadimitriou, and Tim Roughgarden.
\newblock An axiomatic approach to block rewards.
\newblock In {\em Proceedings of the 1st {ACM} Conference on Advances in
  Financial Technologies, {AFT} 2019, Zurich, Switzerland, October 21-23,
  2019}, pages 124--131. {ACM}, 2019.

\bibitem{decker13information}
Christian Decker and Roger Wattenhofer.
\newblock Information propagation in the bitcoin network.
\newblock In {\em 13th {IEEE} International Conference on Peer-to-Peer
  Computing, {IEEE} {P2P} 2013, Trento, Italy, September 9-11, 2013,
  Proceedings}, pages 1--10. {IEEE}, 2013.

\bibitem{ferreira2021dynamic}
Matheus~VX Ferreira, Daniel~J Moroz, David~C Parkes, and Mitchell Stern.
\newblock Dynamic posted-price mechanisms for the blockchain transaction-fee
  market, 2021.

\bibitem{friedenbach17rebatable}
Mark Friedenbach.
\newblock Rebatable fees \& incentive-safe fee markets.
\newblock
  \url{https://lists.linuxfoundation.org/pipermail/bitcoin-dev/2017-September/015093.html},
  2017.

\bibitem{friedenbach18forward}
Mark Friedenbach.
\newblock Forward blocks: On-chain/settlement capacity increases without the
  hard-fork.
\newblock \url{http://freico.in/forward-blocks-scalingbitcoin-paper.pdf}, 2018.

\bibitem{goldberg2006competitive}
Andrew~V. Goldberg, Jason~D. Hartline, Anna~R. Karlin, Michael~E. Saks, and
  Andrew Wright.
\newblock Competitive auctions.
\newblock {\em Games and Economic Behavior}, 55(2):242--269, 2006.

\bibitem{houy2014economics}
Nicolas Houy.
\newblock The economics of bitcoin transaction fees, 2014.

\bibitem{huberman2017monopoly}
Gur Huberman, Jacob~D Leshno, and Ciamac~C Moallemi.
\newblock Monopoly without a monopolist: An economic analysis of the bitcoin
  payment system.
\newblock \url{https://ssrn.com/abstract=3025604}, 2017.

\bibitem{kroll2013economics}
Joshua~A Kroll, Ian~C Davey, and Edward~W Felten.
\newblock The economics of bitcoin mining, or bitcoin in the presence of
  adversaries.
\newblock In {\em Proceedings of WEIS}, volume 2013, 2013.

\bibitem{lange2002impact}
Christine Lange, Christophe Martin, Claire Chabanet, Pierre Combris, and Sylvie
  Issanchou.
\newblock Impact of the information provided to consumers on their willingness
  to pay for champagne: comparison with hedonic scores.
\newblock {\em Food quality and preference}, 13(7-8):597--608, 2002.

\bibitem{lavi19redesigning}
Ron Lavi, Or~Sattath, and Aviv Zohar.
\newblock Redesigning bitcoin's fee market.
\newblock In Ling Liu, Ryen~W. White, Amin Mantrach, Fabrizio Silvestri,
  Julian~J. McAuley, Ricardo Baeza{-}Yates, and Leila Zia, editors, {\em The
  World Wide Web Conference, {WWW} 2019, San Francisco, CA, USA, May 13-17,
  2019}, pages 2950--2956. {ACM}, 2019.

\bibitem{leshno2019bitcoin}
Jacob Leshno and Philipp Strack.
\newblock Bitcoin: An impossibility theorem for proof-of-work based protocols.
\newblock Cowles Foundation Discussion Paper, 2019.

\bibitem{lipmaa2002secure}
Helger Lipmaa, N.~Asokan, and Valtteri Niemi.
\newblock Secure vickrey auctions without threshold trust.
\newblock In Matt Blaze, editor, {\em Financial Cryptography, 6th International
  Conference, {FC} 2002, Southampton, Bermuda, March 11-14, 2002, Revised
  Papers}, volume 2357 of {\em Lecture Notes in Computer Science}, pages
  87--101. Springer, 2002.

\bibitem{liu19survey}
Ziyao Liu, Nguyen~Cong Luong, Wenbo Wang, Dusit Niyato, Ping Wang, Ying-Chang
  Liang, and Dong~In Kim.
\newblock A survey on applications of game theory in blockchain, 2019.

\bibitem{nakamoto08bitcoin}
Satoshi Nakamoto.
\newblock Bitcoin: A peer-to-peer electronic cash system, 2008.

\bibitem{naor99privacy}
Moni Naor, Benny Pinkas, and Reuban Sumner.
\newblock Privacy preserving auctions and mechanism design.
\newblock In Stuart~I. Feldman and Michael~P. Wellman, editors, {\em
  Proceedings of the First {ACM} Conference on Electronic Commerce (EC-99),
  Denver, CO, USA, November 3-5, 1999}, pages 129--139. {ACM}, 1999.

\bibitem{narayanan16bitcoin}
Arvind Narayanan, Joseph Bonneau, Edward~W. Felten, Andrew Miller, and Steven
  Goldfeder.
\newblock {\em Bitcoin and Cryptocurrency Technologies - {A} Comprehensive
  Introduction}.
\newblock Princeton University Press, 2016.

\bibitem{nussair04revealing}
Charles Noussair, Stephane Robin, and Bernard Ruffieux.
\newblock Revealing consumers' willingness-to-pay: A comparison of the bdm
  mechanism and the vickrey auction.
\newblock {\em Journal of Economic Psychology}, 25(6):725 -- 741, 2004.

\bibitem{peter2015transaction}
Peter~R. Rizun.
\newblock A transaction fee market exists without a block size limit.
\newblock Working paper, 2015.

\bibitem{rougharden2020transaction}
Tim Roughgarden.
\newblock Transaction fee mechanism design for the {E}thereum blockchain: An
  economic analysis of {EIP}-1559, 2020.

\bibitem{Sompolinsky15Secure}
Yonatan Sompolinsky and Aviv Zohar.
\newblock Secure high-rate transaction processing in bitcoin.
\newblock In {\em Financial Cryptography and Data Security - 19th International
  Conference, {FC} 2015, San Juan, Puerto Rico, January 26-30, 2015, Revised
  Selected Papers}, pages 507--527. Springer, 2015.

\bibitem{tsabary18gap}
Itay Tsabary and Ittay Eyal.
\newblock The gap game.
\newblock In David Lie, Mohammad Mannan, Michael Backes, and XiaoFeng Wang,
  editors, {\em Proceedings of the 2018 {ACM} {SIGSAC} Conference on Computer
  and Communications Security, {CCS} 2018, Toronto, ON, Canada, October 15-19,
  2018}, pages 713--728. {ACM}, 2018.

\bibitem{wicker2011willingness}
Pamela Wicker.
\newblock Willingness-to-pay in non-profit sports clubs.
\newblock {\em International Journal of Sport Finance}, 6(2):155, 2011.

\bibitem{yao18incentive}
Andrew~Chi{-}Chih Yao.
\newblock An incentive analysis of some bitcoin fee designs, 2018.

\end{thebibliography}

\appendix

\section{The Definition of Strategic Price is Well Formed}
\label{ap:sp_well_formed}

The goal of this section is to show that the infimum in Eq.~\eqref{eq:ps_well_defined} is equal to the minimum. This is shown in Claim~\ref{cl:f_property} below, using $f(b_i) = p^{M}(b_i,b_{-i})$. Claim~\ref{cl:p_monopolistic_satisfies_the_f_property} shows that $f(b_i)$ satisfies the property required by Claim~\ref{cl:f_property}, and clearly, there exists $b_i$ such that $b_i \geq f(b_i)$, e.g., taking any $b_i > \max_{j \neq i} b_j$.

\begin{mclaim}
Let $f:\Re_{\geq 0} \rightarrow \Re_{\geq 0}$ be any function that satisfies the following property: If $x < f(x)$ there exists $\delta >0$ such that for any $0 < \epsilon < \delta$, $f(x + \epsilon) = f(x)$. Let $A = \{ x \in \Re_{\geq 0}\ |\ x \geq f(x)\}$.
Then, if $A$ is not empty, $\inf A \in A$.
\label{cl:f_property}
\end{mclaim}
\begin{proof}
Note that if $A$ is not empty than the infimum of $A$ is well defined since $A$ is bounded from below by $0$. Let $x^* = \inf A$. Assume towards a contradiction that $x^* \notin A$, i.e., $x^* < f(x^*)$. Thus, by the property of $f$, there exists $\epsilon > 0$ such that $x^* + \epsilon < f(x^*)$ and $f(x) = f(x^*)$ for every $x^* < x < x^* + \epsilon$. Since $x^* = \inf A$ there exists $x \in A$, $x^* < x < x^* + \epsilon$. For this $x$, $f(x) \leq x < x^* + \epsilon < f(x^*) = f(x)$, a contradiction.
\end{proof}

\begin{mclaim}
Fix any $\bm b_{-i}$ and any $b_i < p^{M}(b_i, \bm b_{-i})$.
Then there exists $\delta >0$ such that for any $0 < \epsilon < \delta$, $p^{M}(b_i + \epsilon, \bm b_{-i}) =
p^{M}(b, \bm b_{-i})$.
\label{cl:p_monopolistic_satisfies_the_f_property}
\end{mclaim}
\begin{proof}
Let $\bm v = (b_i, \bm b_{-i})$ and assume w.l.o.g.~that $v$ is ordered, i.e., $v_1 \geq v_2 \geq ... \geq v_n$. Let $k$ be such that
$v_k = p^{M}(\bm v)$. Let $j$ be the minimal index of bid $b_i$ in $\bm v$, i.e., $v_{j-1} > v_j = b_i$. Note that $j > k$ since $v_j < v_k$ by assumption. Also note that $v_j \cdot j < v_k \cdot k$ since the monopolistic price is $v_k$. Choose $\delta > 0$ such that $v_j + \delta < v_{j-1}$ and $(v_j + \delta) \cdot j < v_k \cdot k$. By definition, for any $0 < \epsilon < \delta$, $p^{M}(v_j + \epsilon, \bm v_{-j}) = p^{M}(\bm v)$, and the claim follows.
\end{proof}

\section{Auxiliary Technical Claims}
\label{sec:auxiliary_claims}

This section proves multiple useful properties of the monopolistic auction. We recommend that this section be read linearly. Recall that $\mbox{num}(\bm v, z)\equiv |\{v_i|v_i\geq z \}|$.

\begin{mclaim}
If $p^{M}(\bm v_{-i}) \leq v_i$, then $p^{M}(\bm v_{-i}) \leq p^{M}(\bm v)$.
Furthermore, if $p^{M}(\bm v_{-i}) = v_i$, then $p^{M}(\bm v_{-i}) = p^{M}(\bm v)$.
\label{cl:monotonicity_monopolistic_price}
\end{mclaim}
\begin{proof}
Let $x = p^{M}(\bm v_{-i})$ and
$y = p^{M}(\bm v)$. Assume by contradiction that $x>y$.
By definition of the monopolistic price for $\bm v_{-i}$,
$x \cdot \mbox{num}(\bm v_{-i}, x) > y \cdot \mbox{num}(\bm v_{-i}, y)$.
Since $v_i \geq x > y$, $\mbox{num}(\bm v, x) =
\mbox{num}(\bm v_{-i}, x) + 1$ and $\mbox{num}(\bm v, y) =
\mbox{num}(\bm v_{-i}, y) + 1$. Thus,
$x \cdot \mbox{num}(\bm v, x) >
y \cdot \mbox{num}(\bm v, y) - y + x > y \cdot \mbox{num}(\bm v, y)$,
contradicting $y = p^{M}(\bm v)$.

For the second part of the claim, note that the monopolistic price cannot increase when adding $v_i$, otherwise $v_i$ will not be included in the optimal block and therefore it cannot make a difference. The first part of this claim shows that the monopolistic price cannot decrease, hence the equality.
\end{proof}

\begin{mclaim}
For any $\bm v$ and any $i$,
$p^{S}(\bm v_{-i}) < p^{M}(\bm v_{-i})$.
\label{cl:p_strategic_sl_monopolistic}
\end{mclaim}
\begin{proof}
Let $k^* = k^*(\bm v_{-i})$,
and $b^* = \frac{k^*}{k^*+1}p^{M}(\bm v_{-i})$.
Suppose by contradiction that $p^{M}(b^*, \bm v_{-i}) > b^*$. This would imply that $p^{M}(b^*, \bm v_{-i}) = p^{M}(\bm v_{-i})$.
But this is a contradiction since
$\mbox{num}((b^*,v_{-i}),b^*) \cdot b^* \geq (k^* + 1) \cdot b^* =
k^* \cdot p^{M}(\bm v_{-i})$. Thus,
$p^{M}(b^*, \bm v_{-i}) \leq b^*$.
Therefore, $p^{S}(\bm v_{-i})=\min_b p^{M}(b,\bm v_{-i})\leq p^{M}(b^*,\bm v_{-i})\leq b^* < p^M(\bm v_{-i})$, which concludes the proof.
\end{proof}

\begin{cor}
Fix $v_1,...,v_n$ and let $i^* \equiv \mbox{argmax}_{i=1,...,n} v_i$. Then,
\[
p^{S}(\bm v_{-i^*}) < p^{M}(\bm v_{-i^*}) \leq p^{M}(\bm v).
\]
\label{cor:p_strategic_leq_monopolistic}
\end{cor}
This is an immediate consequence of Claims~\ref{cl:p_strategic_sl_monopolistic} and~\ref{cl:monotonicity_monopolistic_price}.

\begin{mclaim}
For any $\bm v$ and any $i$, let
$x = p^{M}(\bm v)$.
Then, $p^{M}(x, \bm v_{-i}) \leq x$
\label{cl:help_help}
\end{mclaim}
\begin{proof}
Suppose towards a contradiction that
$y = p^{M}(x, \bm v_{-i}) > x$.
Therefore $y \cdot \mbox{num}((x, \bm v_{-i}), y) >
x \cdot \mbox{num}((x, \bm v_{-i}), x)$.
Note that
$\mbox{num}(\bm v, y) \geq \mbox{num}((x, \bm v_{-i}), y)$
and
$\mbox{num}(\bm v, x) \leq \mbox{num}((x, \bm v_{-i}), x)$.
This implies\\ $y \cdot \mbox{num}(\bm v, y) \geq
y \cdot \mbox{num}((x, \bm v_{-i}), y) >
x \cdot \mbox{num}((x, \bm v_{-i}), x) \geq
x \cdot \mbox{num}(\bm v, x)$.
Therefore $y \cdot \mbox{num}(\bm v, y) > x \cdot \mbox{num}(\bm v, x)$
which contradicts the fact that $x = p^{M}(\bm v)$.
\end{proof}

\pstrategicalternativedef*
% \begin{mclaim}
% For any $\bm v$ and any $i$,
% $p^{M}(p^{S}(\bm v_{-i}),\bm v_{-i}) =
% p^{S}(\bm v_{-i})$.
% \label{cl:p_strategic_alternative_def}
% \end{mclaim}
\begin{proof}
Recall the definition, $p^{S}(\bm v_{-i}) \equiv
\min \{ b\ |\ p^{M}(b,\bm v_{-i}) \leq b\}$.
It immediately follows that
$p^{M}(p^{S}(\bm v_{-i}),\bm v_{-i}) \leq
p^{S}(\bm v_{-i})$. Assume towards a contradiction that the inequality is strict, and let $x=p^{M}(p^{S}(\bm v_{-i}),\bm v_{-i})$. But then Claim~\ref{cl:help_help} implies that $p^{M}(x,\bm v_{-i}) \leq x$, which is a contradiction since $p^{S}(\bm v_{-i})$ is supposed to be the minimal such $x$.
\end{proof}

Fixing any $\bm v$, recall that $i^* \equiv \mbox{argmax}_{i=1,...,n} v_i$, $k^* \equiv k^*(\bm v_{-i^*}), p^S \equiv p^{S}(\bm v_{-i^*})$.

\xequaly*
% \begin{mclaim}
% Fix any $\bm v$, let $i^* = \mbox{argmax}_{i=1,...,n} v_i$, $k^* = k^*(\bm v_{-i^*})$, $x=p^{M}(\bm v)$ and let $y$ be the smallest element in the support of $F$ which is at least $p^{S}(\bm v_{-i^*})$. Then, $x=y$ implies that $p^{S}(\bm v_{-i^*} ) \geq \frac{k^*}{k^*+1} \cdot p^{M}(\bm v)$.
% \label{cl:x_equal_y}
% \end{mclaim}
\begin{proof}
By Corollary~\ref{cor:p_strategic_leq_monopolistic}, 
\begin{equation}
p^{S}<  p^{M}(\bm v_{-i^*}) \leq p^{M}(v_{i^*}, \bm v_{-i^*}) = p^{M}(\bm v).
\label{eq:p_strategic_leq_p_honest}
\end{equation}
I.e., $y \leq p^{M}(\bm v_{-i^*}) \leq p^{M}(\bm v)$ (because
$p^{M}(\bm v_{-i^*})$ is in the support of $F$). Thus, since $p^{M}(\bm v)=y$, 
\begin{equation}
 p^{M}(\bm v_{-i^*}) = p^{M}(\bm v).
 \label{eq:p_mon_eq_p_honest}
\end{equation}

By definition of $k^*$, there exist $k^*$ values in $\bm v_{-i^*}$ that are at least $p^{M}(\bm v_{-i^*})$. Overall, 
\begin{equation}
\mbox{num}(\bm v_{-i^*},p^{S}) =
\mbox{num}(\bm v_{-i^*},y) = \mbox{num}(\bm v_{-i^*},p^{M}(\bm v)) = k^*.
\label{eq:num_p_strategic}
\end{equation}% and therefore there are $k^*$ in $v_{-i^*}$ that are larger or equal than $p^{S}(v_{-i^*})$.

Therefore,
$R(\bm v_{-i^*},p^{S}) = (k^*+1)p^{S}$: this follows from  Property~\ref{cl:p_strategic_alternative_def}, and Eq.~\eqref{eq:num_p_strategic}. %$p^{M}(p^{S}, v_{-i^*}) = p^{S}$, and from the previous paragraph we know that the number of values in $(p^{S}, v_{-i^*})$ that are larger or equal than $p^{S}$ is $k^*+1$.
Thus, for any number $z$,
%for any number $z > p^{S}$,
\begin{equation}
    (k^*+1)p^{S}
\geq
z \cdot \mbox{num}((\bm v_{-i^*}, p^S), z) \geq
z \cdot \mbox{num}(\bm v_{-i^*}, z).
\label{eq:k_star_plus_1_times_ps}
\end{equation}
%(since $z > p^{S}$, the number of bids in $(p^{S}, v_{-i^*})$ that are at least $z$ is exactly $\mbox{num}(\bm v_{-i^*}, z)$).

Take $z\equiv p^{M}(\bm v)$. %which satisfies $z > p^{S}$ by Eq.\eqref{eq:p_strategic_leq_p_honest}.
Note that $\mbox{num}(\bm v_{-i^*}, p^{M}(\bm v)) = k^*$ using Eq.~\eqref{eq:p_mon_eq_p_honest}. Using Eq.~\eqref{eq:k_star_plus_1_times_ps},
%that $p^{M}(\bm v)=y$ implies that
$(k^*+1)p^{S} \geq  p^{M}(\bm v) \cdot k^* $, and the claim follows.
\end{proof}

\begin{mclaim}
For any $\bm v = (v_1,...,v_n)$, and any $v'_i$ such that
$v_i \geq v'_i \geq p^{M}(\bm v)$,
$p^{M}(\bm v) \geq p^{M}(v'_i, \bm v_{-i})$.
\label{cl:p_monopolistic_is_monotone}
\end{mclaim}
We remark that $p^{M}(\bm v)$ is not necessarily monotonically increasing as a function of $v_i$, for example, $p^{M}(2,0)=2>1=p^{M}(2,1)$. Hence, the importance of the condition $v'_i \geq p^{M}(\bm v)$.
\begin{proof}
Denote $x = p^{M}(\bm v)$ and
$y = p^{M}(v'_i, \bm v_{-i})$.
Assume towards a contradiction that $x <  y$.
Therefore
\begin{align*}
y \cdot \mbox{num}(\bm v, y) \geq
y \cdot \mbox{num}((v'_i, \bm v_{-i}), y) \\
> x \cdot \mbox{num}((v'_i,\bm v_{-i}), x) =
x \cdot \mbox{num}(\bm v, x)
\end{align*}
where the first step follows since $v_i \geq v'_i$,
the second step follows since $y$ is the monopolistic
price for $(v'_i, \bm v_{-i})$ (the inequality is strict because $y>x$), and the third step follows since $v_i \geq v'_i \geq x$. However
$y \cdot \mbox{num}(\bm v, y) > x \cdot \mbox{num}(\bm v, x)$ contradicts
the fact that $x$ is the monopolistic price for $v$.
\end{proof}

\begin{mclaim}
For any $\bm v = (v_1,...,v_n)$, and any $v'_i$ such that
$v_i \geq v'_i \geq p^{M}(v'_i, \bm v_{-i})$,
$p^{M}(\bm v) \geq p^{M}(v'_i, \bm v_{-i})$.
\label{cl:p_mon_is_monotone_1}
\end{mclaim}
\begin{proof}
Assume towards a contradiction that $p^{M}(\bm v) < p^{M}(v'_i, \bm v_{-i})$. But then
$v_i \geq v'_i \geq p^{M}(\bm v)$ and
Claim~\ref{cl:p_monopolistic_is_monotone} implies that
$p^{M}(\bm v) \geq p^{M}(v'_i, \bm v_{-i})$ which is a contradiction.
\end{proof}

\begin{mclaim}
For any $\bm v = (v_1,...,v_n)$, and any $i$, if $v_i < p^{M}(\bm v)$ then $v_i < p^{S}(\bm v_{-i})$.
\label{cl:loser_cannot_gain}
\end{mclaim}
\begin{proof}
Denote $x \equiv p^{M}(\bm v)$ and
$y \equiv p^{S}(\bm v_{-i}) = p^{M}(y, \bm v_{-i})$ (according to Property~\ref{cl:p_strategic_alternative_def}). Assume towards a contradiction that $x > v_i \geq y$. 
Since $x \equiv p^{M}(\bm v)$ and our assumption $y<x$,
\begin{equation}
x \cdot \mbox{num}(\bm v, x) > y \cdot \mbox{num}(\bm v, y)    
\label{eq:xnumx_gr_ynumy}
\end{equation} 
Furthermore, since $y = p^{M}(y, \bm v_{-i})$),
\begin{equation}
y \cdot \mbox{num}((y, \bm v_{-i}), y) \geq x \cdot \mbox{num}((y, \bm v_{-i}), x).  
\label{eq:y_nym_v_min_i_y_geq}
\end{equation}
 
We also have 
\begin{equation}
\mbox{num}(\bm v, y) = \mbox{num}((y, \bm v_{-i}), y), 
\label{eq:num_v_y_eq_num_v_min_i_y}
\end{equation}
since $v_i \geq y$, and similarly,
\begin{equation}
 \mbox{num}(\bm v, x)= \mbox{num}((y, \bm v_{-i}), x) ,
 \label{eq:num_v_x_eq_num_y_v_min_i_x}
\end{equation}
since $v_i < x$ and $y<x$.

By plugging in Eqs.~\eqref{eq:num_v_y_eq_num_v_min_i_y} and~\eqref{eq:num_v_x_eq_num_y_v_min_i_x} in Eq.~\eqref{eq:y_nym_v_min_i_y_geq},
\begin{equation}
   y \cdot \mbox{num}(\bm v, y) \geq x \cdot \mbox{num}(\bm v, x). 
   \label{eq:y_num_v_y_geq_x_num_v_x}
\end{equation}
Combining Eqs.~\eqref{eq:xnumx_gr_ynumy} and \eqref{eq:y_num_v_y_geq_x_num_v_x}, we finally show that $x \cdot \mbox{num}(\bm v, x) > x \cdot \mbox{num}(\bm v, x)$, which proves the  contradiction.
\end{proof}

\begin{mclaim}
\label{cl:mon_for_win_st_p}
For any $v_1,...,v_n$ and $i,j$,
$v_i > v_j \geq p^{S}(\bm v_{-j})$ implies that
$p^{S}(\bm v_{-j}) \geq p^{S}(\bm v_{-i})$.
\end{mclaim}
We remark that the condition $v_j \geq p^{S}(\bm v_{-j})$
is important as without it the claim is not true. For example,
take $\bm v = (2,1,0)$:
\[ p^{S}(\bm v_{-2}) = p^{S}(2,0) = 1 > \frac{2}{3} = p^{S}(2,1)=p^{S}(\bm v_{-3}).\]
\begin{proof}
Let $b = p^{S}(\bm v_{-j})$.
We will prove that $b \geq p^{M}(b,\bm v_{-i})$, which
implies the claim since $p^{S}(\bm v_{-i}) \equiv
\min \{b\ |\ b \geq p^{M}(b,\bm v_{-i})\}$.

By applying Claim~\ref{cl:p_monopolistic_is_monotone} with respect to $\bm u=(v_1,\ldots,v_{j-1},b,v_{j+1},\ldots,v_n)$ and
\[\bm u'=(v_1,\ldots,v_{i-1},v_j,v_{i+1},\ldots, v_{j-1},b,v_{j+1},\ldots,v_n)\]
implies that
$p^{M}(\bm u) \geq
p^{M}(\bm u')$
%$p^{M}(b, v_i, v_{-i, -j}) \geq p^{M}(b, v_j, v_{-i,-j})$ by taking the vector of values $(b, v_{-j})$ and $v'_i = v_j$ 
(note that $v_i \geq v_j \geq b \geq p^{M}(b,\bm v_{-j})=p^{M}(\bm u)$ where the last inequality follows from the definition of $p^{S}(\bm v_{-j})$; therefore $u_i \geq u'_i \geq p^{M}(\bm u)$). Thus
\[
b \geq p^{M}(\bm u) \geq p^{M}(\bm u') = p^{M}(b, \bm v_{-i}).
\]
The last equality holds since the monopolistic price is oblivious to the identities of the users. The claim thus follows.
\end{proof}

\deltabiggerforlargerbidders*
% \begin{mclaim}
% For any $v_1,...,v_n$ and $i,j$, 
% $v_i \geq v_j$ implies $\delta_i(v_i,\bm v_{-i}) \geq \delta_j(v_j,\bm v_{-j})$. 
% \label{cl:delta_bigger_for_larger_users}
% \end{mclaim}
\begin{proof}
If $v_j < p^{S}(\bm v_{-j})$, $\delta_j(v_j,\bm v_{-j})=0$ and the claim follows. Thus assume $v_j \geq p^{S}(\bm v_{-j})$, and we have $v_i > v_j \geq p^{S}(\bm v_{-j}) \geq p^{S}(\bm v_{-i})$ where the last inequality follows from Claim~\ref{cl:mon_for_win_st_p}.
Since $p^{M}$ is the same in both terms,
and $p^{S}(\bm v_{-j}) \geq p^{S}(\bm v_{-i})$, we have $\delta_i(\bm v_i,v_{-i}) \geq \delta_j(\bm v_j,v_{-j})$.
\end{proof}

\begin{mclaim}
For any $\bm v_{-i}$ and $v_i > v'_i$,
$\delta_i(v_i,\bm v_{-i}) \geq \delta_i(v'_i,\bm v_{-i})$. 
\label{cl:delta_is_monotone}
\end{mclaim}
\begin{proof}
If $v'_i < p^{S}(\bm v_{-i})$, $\delta_i(v'_i,\bm v_{-i})=0$ and the claim follows. Thus assume $v'_i \geq p^{S}(\bm v_{-i})$. Assume towards a contradiction that $v'_i < p^{M}(v'_i, \bm v_{-i})$. Claim~ \ref{cl:loser_cannot_gain} then implies that $v'_i < p^{S}(\bm v_{-i})$ which is a contradiction to our assumption above. Thus, 
$v_i > v'_i \geq p^{M}(v'_i, \bm v_{-i})$. Claim~\ref{cl:p_mon_is_monotone_1} then implies that $ p^{M}(v_i, \bm v_{-i}) \geq p^{M}(v'_i, \bm v_{-i})$. By the definition of the discount ratio we have $\delta_i(v_i,\bm v_{-i}) \geq \delta_i(v'_i,\bm v_{-i})$ as claimed.
\end{proof}

\section{Additional Auxiliary Claims for the Multi-Bid Case}
\label{sec:multibid_auxiliary_claims}

\begin{mclaim}
For any $\bm b_{-i}$ and $u \in \N^+$, define\footnote{The minimum is well defined -- see the argument next to Eq.~\eqref{eq:p_multibid}.}
\[
p_u^{multibid}(\bm b_{-i})=\min \{u \cdot b_i\ |\ b_i\in \R,\ b_i\geq p^{M}(\overbrace{b_i,\ldots,b_i}^{\text{u times}},\bm b_{-i})\}.
\]
Then, for any $u \geq n+1$, $p_u^{multibid}(\bm b_{-i}) > p_1^{multibid}(\bm b_{-i})$.
\label{cl:u_cannot_be_large}
\end{mclaim}
\begin{proof}
Let $b = \frac{p_u^{multibid}(\bm b_{-i})}{u}$. Then,
$b \cdot (u + \mbox{num}(\bm b_{-i}, b)) \geq k^*(\bm b_{-i}) \cdot p^{M}(\bm b_{-i})$. Therefore,
\begin{equation}
p^{multibid}_u(\bm b_{-i}) = b \cdot u
\geq \frac{k^*(\bm b_{-i}) \cdot p^{M}(\bm b_{-i})}{u + \mbox{num}(\bm b_{-i}, b)} \cdot u
\label{eq:u_for_multibid}
\end{equation}
If $k^*(\bm b_{-i}) \geq 2$ then $\frac{u}{u + \mbox{num}(\bm b_{-i}, b)} \cdot k^*(\bm b_{-i}) > 1$ since $u > n > \mbox{num}(\bm b_{-i}, b)$. Therefore Eq.~\eqref{eq:u_for_multibid} implies that $p^{multibid}_u(\bm b_{-i}) > p^{M}(\bm b_{-i})$. Using Claim~\ref{cl:p_strategic_sl_monopolistic}, $ p^{M}(\bm b_{-i})> p^{S}(\bm b_{-i})$, and since $p^{S}(\bm b_{-i}) = p^{multibid}_1(\bm b_{-i})$, we complete the proof in the case $k^*(\bm b_{-i})\geq 2$.

If $k^*(\bm b_{-i}) = 1$ then $\frac{u}{u + \mbox{num}(\bm b_{-i}, b)} \cdot k^*(\bm b_{-i}) > \frac{1}{2}$. Therefore Eq.~\eqref{eq:u_for_multibid} implies that $p^{multibid}_u(\bm b_{-i}) > \frac{p^{M}(\bm b_{-i})}{2} \geq p^{S}(\bm b_{-i}) = p^{multibid}_1(\bm b_{-i})$, implying the claim. (Here, in the second inequality we used the fact that $k^*(\bm b_{-i})=1$, and bidding half of the highest bid will win in this case.) 
\end{proof}

\begin{mclaim}
Fix any $\bm b_{-i}$ and let $b,u$ be the arguments which minimize Eq.~\eqref{eq:p_multibid}.
Then,
\[
b = p^{M}(\overbrace{b,\ldots,b}^{\text{u times}},\bm b_{-i}).
\]
\label{cl:p_multibid-equality}
\end{mclaim}
\begin{proof}
By definition
$b \geq p^{M}(\overbrace{b,\ldots,b}^{\text{u times}},\bm b_{-i})$.
Suppose towards a contradiction that the inequality is strict, and take some $b'$ such that
\[
b > b' \geq p^{M}(\overbrace{b,\ldots,b}^{\text{u times}},\bm b_{-i}).
\]
Claim~\ref{cl:p_monopolistic_is_monotone} implies that
\[
b' \geq p^{M}(\overbrace{b',\ldots,b'}^{\text{u times}},\bm b_{-i}),
\]
contradicting the minimality of $(b,u)$ since $b \cdot u > b' \cdot u$.
\end{proof}

\begin{mclaim}
Fix any $\bm b_{-i}$ and let $(b,u)$ be the arguments which minimize Eq.~\eqref{eq:p_multibid}. Then,
\[
b \leq p^{M}(\bm b_{-i}).
\]
\label{cl:k_start_multibid}
\end{mclaim}
\begin{proof}
Let $b' = p^{M}(\bm b_{-i})$. Then, the second part of Claim~\ref{cl:monotonicity_monopolistic_price} implies that $b' = p^{M}(b', \bm b_{-i})$. Thus $(b',u'=1)$ satisfies the requirement of Eq.~\eqref{eq:p_multibid}. Hence $b \cdot u \leq b'$ and the claim follows.
\end{proof}

\section{Properties of the binomial and trinomial distributions}
%\section{A few mathematical facts}
\label{ap:binomial_trinomial}

\begin{mclaim}
Let $X\sim B(n,p)$ be a binomial random variable where $0<p<1$. For every $i$, $\Pr(X=i)=O\left(\frac{1}{\sqrt n}\right)$. 
\label{cl:binomial}
\end{mclaim}

\begin{proof}
If $np$ is an integer, the mode (most likely value of the distribution) is its mean $np$.\footnote{This can be shown directly by computing $\frac{\Pr(X=i)}{\Pr(X=i+1)}$.} The case where $np$ is not an integer can be handled is a similar manner.
Therefore, our goal is to show that $\Pr(X=np)=O\left(\frac{1}{\sqrt n}\right)$. By using Stirling's approximation\footnote{One could take into account the error term, i.e., $n!= \sqrt{2 \pi n} \left( \frac{n}{e}\right)^n(1+O(\frac{1}{n}))$, and reach the same conclusion.} $n!\approx \sqrt{2 \pi n} \left( \frac{n}{e}\right)^n$,
\begin{align*}
\Pr(X=np)&=\binom{n}{np}p^{np}(1-p)^{n(1-p)}\\
&\approx \frac{\sqrt{2 \pi n} (\frac{n}{e})^n }{\sqrt{2 \pi np} (\frac{np}{e})^{np} \sqrt{2 \pi n(1-p)} (\frac{n(1-p)}{e})^{n(1-p)} } p^{np}(1-p)^{n(1-p)} \\
&=\frac{1}{\sqrt{2 \pi n p (1-p)}} = O\left(\frac{1}{\sqrt{n}}\right).
\end{align*}
Of course, this is only valid when $p$ is treated as a constant.
\end{proof}

\begin{mclaim}
Let $(X_1,X_2,X_3)\sim Trinomial(n,p_1,p_2,p_3)$ be a triple of trinomial random variables where $p_1,p_2 > 0$. For every function $f$, $\Pr(X_1=f(X_1+X_2))=O\left(\frac{1}{\sqrt n}\right)$.  
\label{cl:trinomial}
\end{mclaim}
Recall that the trinomial distribution is the multinomial distribution restricted to the case $k=3$.
\begin{proof}

Our first goal is to prove a very weak bound $\Pr(X_3\geq \alpha n) = O(\frac{1}{\sqrt{n}})$ for $\alpha =\frac{p_3 + 1}{2}$. Since this bound is so weak it could be proved using the Chrnoff bound, or even Chebyshev's inequality. \begin{lemma}[Chebyshev's inequality]
For every $k\in \R$, and every random variable $X$ with expectation $\mu$ and standard deviation $\sigma$,
\begin{equation}
    \Pr(|X-\mu|\geq k \sigma) \leq \frac{1}{k^2}
\end{equation}
\label{le:chebyshev}
\end{lemma}
Recall that $X_3\sim Binomial(n,p_3)$, and therefore, $E[X]=np_3$ and has a standard deviation $\sqrt{n p_3(1-p_3)}$.
By observing that $ \Pr(X-\mu\geq k \sigma) \leq \Pr(|X-\mu|\geq k \sigma)$ and plugging in $k=\frac{\sqrt{n (1-p_3)}}{2\sqrt{p_3}}$ in Lemma~\ref{le:chebyshev},
\begin{equation}
\Pr(X_3\geq \alpha n) \leq \frac{4p_3}{n(1-p_3)}=O\left(\frac{1}{\sqrt n}\right),
\label{eq:x_3_geq_alpha_n}
\end{equation}
where in the equality we use the fact that $p_3<1$ (since $p_1,p_2> 0$ and $p_1+p_2 +p_3=1$).

Conditioned that $X_3=z$, $X_1$ has a Binomial distribution: $X_1\sim B(n-z,\frac{p_1}{1-p_3})$ (here we use the condition  $p_1,\ p_2 > 0$ to conclude that $0 < \frac{p_1}{1-p_3} <1$). Therefore,
\begin{align*}
 \Pr&(X_1=f(X_1+X_2))=\sum_{z=0}^n \Pr(X_1=f(X_1+X_2)|X_3=z) \Pr(X_3=z)\\
 &=\sum_{z=0}^n \Pr(X_1=f(n-z)|X_3=z) \Pr(X_3=z)\\
  &\leq \sum_{z=0}^{\floor{\alpha n}} \Pr(X_1=f(n-z)|X_3=z) \Pr(X_3=z) + \sum_{z=\floor{\alpha n}+1}^{n} \Pr(X_3=z)\\
 & \leq \sum_{z=0}^{\floor{\alpha n}} \Pr(X_1=f(n-z)|X_3=z) \Pr(X_3=z) + O\left(\frac{1}{\sqrt{n}}\right) \\
  & \leq \sum_{z=0}^{\floor{\alpha n}} O\left(\frac{1}{\sqrt{n-z}}\right) \Pr(X_3=z) + O\left(\frac{1}{\sqrt{n}}\right) \\
   & \leq \sum_{z=0}^{\floor{\alpha n}} O\left(\frac{1}{\sqrt{(1-\alpha)n}}\right) \Pr(X_3=z) + O\left(\frac{1}{\sqrt{n}}\right) =O\left(\frac{1}{\sqrt{n}}\right),
\end{align*} 
where in the second inequality we use Eq.~\eqref{eq:x_3_geq_alpha_n}, and Claim~\ref{cl:binomial} in the third inequality.
\end{proof}

\begin{mclaim}
Let $F$ be any distribution with a finite support size and $\bm v = (v_1,..v_n)$ be $n$ i.i.d.~draws from $F$. Define \nom{num}{$\mbox{num}(\bm v, z)$}{$\mbox{num}(\bm v, z) \equiv |\{v_i|v_i\geq z \}|$} $\equiv |\{v_i|v_i\geq z \}|$ and random variables $n_z = \mbox{num}(\bm v,z)$ for any real number $z$. Then, for any arbitrary $x,y\in \mbox{Support}(F)$ such that $x>y$,
$$
\lim_{n \rightarrow \infty} \Pr(n_x \cdot x > n_y \cdot y \geq (n_x - 1)x ) = 0.
%\label{eq:pr_h_g}
$$
\label{cl:B_has_zero_probability}
\end{mclaim}
\begin{proof}
The term $n_x \cdot x > n_y \cdot y \geq  (n_x - 1)x$ is the same as $n_x > \frac{y}{x} n_y \geq n_x - 1$, which is true if and only if $n_x = \lfloor{\frac{y}{x}n_y+1}\rfloor$. The triple $(n_x, n_y - n_x, n - n_y)$ is a trinomial distribution with probabilities $p_1 = \Pr_{v_i \sim F}(v_i \geq x), p_2 = \Pr_{v_i\sim F}(x > v_i \geq y), p_3 = \Pr_{v_i\sim F}(v_i < y)$ ($p_i$ depends on $F$ but not on $n$). Denoting $f(n_y) = \floor{\frac{y}{x}n_y+1}$, we conclude that $\Pr(n_x=f(n_y)) = \Pr(n_x \cdot x > n_y \cdot y \geq (n_x - 1)x)$.
Claim~\ref{cl:trinomial} therefore implies the current claim.
\end{proof}

\section{Proof of Theorem~\ref{thm:bne-vs-discount-ratio}}
\label{sec:proof-of-thm-bne-vs-discount-ratio}

%\begin{proof}
We first show a lemma:

\begin{lemma}
\label{p_dif_geq_u_dif}
For every $v_1,b_1,\bm v_{-1}$,
\[p^{M}(v_{max},\bm v_{-1}) - p^{S}(\bm v_{-1}) \geq u_1(v_1, b_1, \bm v_{-1}) - u_1(v_1, v_1, \bm v_{-1})
\]
\end{lemma}
\begin{proof}
Assume first that $v_1 \geq p^{M}(v_1, \bm v_{-1})$. Then:
\begin{align*}
p^{M}(v_{max},\bm v_{-1}) - p^{S}(\bm v_{-1}) & \geq p^{M}(v_1,\bm v_{-1}) - p^{S}(\bm v_{-1}) \\
& = [v_1 - p^{S}(\bm v_{-1})] - [v_1 - p^{M}(v_1,\bm v_{-1}))]\\
& = u_1(v_1, p^{S}(\bm v_{-1}), \bm v_{-1}) - u_1(v_1, v_1, \bm v_{-1})\\
& \geq u_1(v_1, b_1, \bm v_{-1})] - u_1(v_1, v_1, \bm v_{-1})
\end{align*}
where the first transition follows by Claim~\ref{cl:p_mon_is_monotone_1} and the third transition follows by Property~\ref{cl:p_strategic_alternative_def}. Therefore if $v_1 \geq p^{M}(v_1, \bm v_{-1})$ the claim holds. Now assume $v_1 < p^{M}(v_1, \bm v_{-1})$. In this case, $u_1(v_1, v_1, \bm v_{-1}) = 0$ and
$u_1(v_1, b_1, \bm v_{-1}) \leq 0$ (since by Claim~\ref{cl:loser_cannot_gain} in this case
$v_1 < p^{S}(\bm v_{-1})$). Thus,
$p^{M}(v_{max},\bm v_{-1})) - p^{S}(\bm v_{-1}) \geq 0 \geq u_1(v_1, b_1, \bm v_{-1}) - u_1(v_1, v_1, \bm v_{-1})$ hence the claim holds in this case as well.
\end{proof}

We now prove the theorem:

\begin{align*}
\Delta_n^{GT} &=  \E_{\bm v_{-1} \sim F} \left[\max_{v_1 \in Support(F)} \delta_1(v_1, \bm v_{-1}) \right] \\
&= \E_{\bm v_{-1} \sim F} \left[\delta_1(v_{max}, \bm v_{-1}) \right]\\
&= \E_{\bm v_{-1} \sim F} \left[1-\frac{p^{S}(\bm v_{-1})}{p^{M}(v_{max},\bm v_{-1})} \right]\\
&\geq \frac{1}{v_{max}}\E_{\bm v_{-1} \sim F} \left[p^{M}(v_{max},\bm v_{-1})-p^{S}(\bm v_{-1}) \right]\\
&\geq \max_{v_1,b_1 \in Support(F)} \left \{ \frac{1}{v_{max}}\E_{\bm v_{-1} \sim F} \left[u_1(v_1,b_1,\bm v_{-1} )-u_1(v_1,v_1,\bm v_{-1}) \right] \right \}.
\end{align*}
Here the second step follows from Claim~\ref{cl:delta_is_monotone}, the third step follows since a user with the highest possible valuation $v_{max}$ will necessarily win, the forth step follows from $p^M(v_{max},\bm v_{-1})\leq v_{max}$, and the fifth step follows from Lemma~\ref{p_dif_geq_u_dif}. Eq.~(\ref{eq:BNE}) directly follows for $\epsilon = v^{max} \cdot \Delta^{GT}_n$, which completes the proof of the theorem.
%\end{proof}

\section{Analysis of Example~\ref{ex:2n1n}}
\label{sec:analysis_RSOP_example}

\begin{mclaim}
\begin{equation}
\lim_{n\rightarrow \infty} \frac{RSOP(\overbrace{2,\ldots,2}^{n\text{ times}},\overbrace{1,\dots,1}^{n\text{ times}})-2n}{\sqrt n }=-\frac{1}{\sqrt{\pi}}
\end{equation}
\label{cl:lim_rsop}
\end{mclaim}
\begin{proof}
Let $X_i$ be the indicator variable whether the $i$'th 2-bidder is in group $A$, and similarly $Y_i$ for the $i$'th $1$-bidder. Let $Z_i=X_i-Y_i$. Let $X=\sum_{i=1}^n X_i$, $Y=\sum_{i=1}^n Y_i$ and $Z=\sum_{i=1}^n Z_i$.

Note that the revenue for any realization of the coin toss is exactly $2n - |Z|$. Suppose first that $Z \geq 0$, i.e., $X \geq Y$. In this case, the price offered to set $A$ (which is the monopolistic price computed on set B) is 1, and similarly the price offered to set $B$ is $2$. Thus, the revenue is $1 \cdot (X+Y) + 2 \cdot (n-X) = 2n - Z$. Similarly, if $Z<0$, the revenue is $2n + Z$. This shows that in all cases, the revenue is $2n-|Z|$.

Since $\E[Z_i]=0$ and $\Var(Z_i)=\frac{1}{2}$, by the central limit theorem, $\frac{1}{\sqrt n} Z$ converges in distribution to $N(0,\sigma^2=\frac{1}{2})$ as $n$ goes to infinity.
Recall that for a normal random variable $X'\sim N(0 , \sigma)$, the distribution of $Y'=|X'|$ is called half-normal, and $E[Y']=\frac  {\sigma {\sqrt 2 }}{\sqrt \pi }$.
Therefore, $E[\frac{|Z|}{\sqrt{n}}]=\frac{\sigma \sqrt 2}{\sqrt \pi}=\frac{1}{\sqrt \pi}$. This completes the proof, since $RSOP(\overbrace{2,\ldots,2}^{n\text{ times}},\overbrace{1,\dots,1}^{n\text{ times}})=2n-\E[|Z|]$.
\end{proof}
An immediate consequence of Claim~\ref{cl:lim_rsop} is the following:
\begin{cor}\label{cor:RSOP_example}
$RSOP(\overbrace{2,\ldots,2}^{n\text{ times}},\overbrace{1,\dots,1}^{n\text{ times}}) = 2n-\sqrt{\frac{n}{\pi}} + o(\sqrt{n})$.
\end{cor}

\section{Desiderata}
\label{sec:desiderate}
In this appendix we provide the desired properties of a fee mechanism in Bitcoin. We note that some of the properties mentioned below may conflict with others, and it may well be that no single mechanism provides them simultaneously. 

\paragraph{High social welfare} The sum of valuations of accepted transactions should be large. In this sense, the system is providing a high level of utility to its users. In particular, an efficient allocation is warranted: transactions with higher values should be included before those with lower values (assuming both are available at the same time) as this clearly improves the social welfare.

\paragraph{Revenue extraction} The amount of money transferred to  the miners is high, which would buy more security for the system. Stated another way: the revenue of miners should be part of the social welfare measure of the system.

\paragraph{Truthful bidding} An ideal mechanism would allow users to state their preferences clearly and would encourage honest reporting. The main advantage would be that users do not need to ``overthink'' how they should act. Furthermore, this removes computational burdens associated with such strategies: there is no need to monitor the network to obtain statistics on past bids, perform estimates of the threshold bid needed to get accepted into a block, and no need to perform bid updates for the transaction. 
Note that truthfulness should imply that various, more technical, manipulations, should not be profitable. 
For example, such manipulations include splitting a single transaction into two transfers of smaller amounts, adding more transactions between two bitcoin addresses of the same client, etc. A desired property that such manipulation would not be profitable. 

\paragraph{Resistance to manipulation by miners} 
The mechanism should be resistant to selfish behavior by the miners. Such behavior can include miners adding transactions of their own into their own blocks, miners withholding transactions and selecting other sets of transactions, etc. 

\paragraph{Resistance to manipulation via side payments} Some mechanisms may appear well at first sight, but in fact may encourage miners and users to circumvent the fee system altogether. For example, if a miner has to give half of the revenue from his block to the miner that creates the next block in the chain, then he can offer users a deal: The miner will include transactions with 0 fees in his block in exchange for a side payment that will be given to him via a separate and direct transaction. In this manner, he does not need to share the rewards with others. Of course, since Bitcoin provides a natural payment channel from the users to the miners, the mechanism should be resistant to such side payments.  

\paragraph{Adaptivity to changing demand, network conditions, block sizes, etc.} 
It is important to avoid hard-coded magic numbers (such as hard-coded minimal fees) in the protocol as much as possible. Given that the protocol is hard to update, as it requires wide adoption of new code, any hard-coded number is difficult to adjust. A minimal fee, for example, which may need adjusting from time to time (e.g., when the exchange rate fluctuates, or demand increases) would insert inefficiency when it is not set in accordance with market conditions.

\paragraph{Accounting for temporal considerations} Users may have different levels of urgency for their transactions. A good fee mechanism will take into account the willingness of users to delay the acceptance of their transactions, e.g., in exchange for paying lower fees. In Bitcoin, transactions that were not added to the blockchain persist with the same bid and may be included in other blocks (without a discount). Our own analysis in this paper considers only impatient users who only desire that their transactions will be included in the next block. We believe that such temporal considerations are an important issue to tackle in future work.

\paragraph{Consensus based mechanism} One of the key features of cryptocurrencies such as Bitcoin is that it \emph{can} be assumed that data that is part of the longest blockchain is agreed upon all miners and users. Therefore, the mechanism could make decisions based on previous blocks. On the other hand, there is no consensus between the miners regarding data which is not part of the blockchain. For example, different miners might have different views regarding the mempool (which is the set of transaction waiting to be included in the next block). In particular, it is impossible to take the statistics of \emph{all} bids in previous rounds into account, since the bids that were not included in the block are not part of the consensus.

\section{Nomenclature}
\label{sec:nomenclature}

\vspace{5mm}

\renewcommand{\nomname}{}\vspace{-10pt}%a fix to remove the by default nomenclature paragraph heading.
%\nom is a new command declared in macros.tex which takes 3 arguments: the first argument to sort, the second for the name and third for the description and prints the second argument there.
\printnomenclature[1in] %Removed the "Nomenclature" title, since this is already the heading of the appendix

\end{document}